\documentclass[
 superscriptaddress,
nofootinbib,
 amsmath,amssymb,
 twocolumn,
 aps,
 prd,
longbibliography
]{revtex4-2}

\usepackage{amsthm}

\usepackage[utf8]{inputenc}
\usepackage{geometry, cancel}
\usepackage{float}
\usepackage{graphicx,xcolor}
\usepackage{enumitem}
\usepackage{braket}
\usepackage[colorlinks=true]{hyperref}
\usepackage{amsfonts,amssymb}
\usepackage{amsmath, mathtools}
\usepackage{caption,subcaption}

\usepackage[normalem]{ulem}

\newtheorem{lemma}{Lemma}
\newtheorem{theorem}{Theorem}

\newcommand{\Tr}{\textrm{Tr}}

\begin{document}

\title{Detector-based measurements for QFT: two issues and an AQFT proposal}
\author{Nicola Pranzini}
\email{nicola.pranzini@helsinki.fi}
\address{Department of Physics, P.O.Box 64, FIN-00014 University of Helsinki, Finland}
\address{QTF Centre of Excellence, Department of Physics, University of Helsinki, P.O. Box 43, FI-00014 Helsinki, Finland}
\address{InstituteQ - the Finnish Quantum Institute, Finland}
\author{Esko Keski-Vakkuri}
\email{esko.keski-vakkuri@helsinki.fi}
\address{Department of Physics, P.O.Box 64, FIN-00014 University of Helsinki, Finland}
\address{InstituteQ - the Finnish Quantum Institute, Finland}
\address{Helsinki Institute of Physics, P.O.Box 64, FIN-00014 University of Helsinki, Finland}

\begin{abstract}
    We present and investigate two issues within the measurement scheme for QFT presented by J. Polo-G\'omez, L. J. Garay and E. Mart\'in-Mart\'inez in “A detector-based measurement theory for quantum field theory"~\cite{Polo-GomezEtAl21}. We point out some discrepancies that arise when the measurement scheme is applied to contextual field states and show that $n$-point function assignments based on local processing regions sometimes lead to inconsistencies. To solve these issues, we modify the measurement scheme to use non-relativistic detectors to induce an update rule for algebraic states in the Haag-Kastler formulation of quantum field theory. In this way, $n$-point functions can be consistently evaluated across any region having a definite causal relation with measurements.
\end{abstract}

\maketitle

\section{Introduction}

In Ref.~\cite{Polo-GomezEtAl21}, J. Polo-G\'omez, L. J. Garay, and E. Mart\'in-Mart\'inez (henceforth, PGM) provided a measurement scheme for quantum fields through Unruh-DeWitt detectors as probes. The idea is to test a quantum field via a non-relativistic system used as a detector, applying the measurement postulate of QM to the latter to measure the former. Measurements introduced in this way are compatible with relativity and safe from the causality violations originally presented by R. D. Sorkin in \cite{Sorkin93}. The framework accounts for quantum correlations between spacelike-separated regions of spacetime, consistently with the standard no-communication theorem of QM~\cite{PeresT04, FlorigS97}. For this purpose, field states are taken to be contextual\footnote{The adjective \textit{contextual} has different interpretations in different settings. Throughout this article, we will use it in the same sense as PGM used it, i.e. as a synonym of \textit{subjective}.} to observers, i.e. they are treated as states of knowledge of observers instead of as elements of reality. Finally, the authors propose that all the information about the field can be specified in terms of $n$-point functions computed from the contextual field states, as tehy fully characterise the theory without giving explicit spacetime dependence to quantum field states.

We suggest modifying the framework proposed by PGM by interpreting states as non-contextual functions of spacetime regions. This modification allows addressing two questions arising from PGM's proposal, namely: how to describe measurements performed on an already collapsed state (discussed in the subsection \ref{s.issue1} of this introduction), and how to render PGM's $n$-point functions assignment consistent for all point configurations (discussed in the subsection \ref{s.issue2} of this introduction). As anticipated, answering these questions requires modifying PGM's framework. Specifically, we propose introducing spacetime-dependent (equivalence classes of) quantum field states, which address the first question (see Sec.~\ref{s.density states}), and relate them to algebraic states arising from algebraic QFT (AQFT), which answer the second question (see Sec.~\ref{s.algebraic states}). The result is a framework similar to that of K. -E. Hellwig and K. Kraus \cite{HellwigK70}, but different for the use of non-relativistic probes, their relation with the algebraic operation induced on the field algebra, and the specific algebraic state assignment. In case one is interested in the information that can be extracted locally and non-selectively, PGM's states can be recovered by looking at the maximal information available for an observer placed in the spacetime region where the state is evaluated on.

Our work aligns with the idea of addressing measurement-related issues via AQFT. This concept appeared both in PGM's work and in an article by C. Fewster and R. Verch (FV), who introduced a novel approach to measure a quantum field employing another field as a probe via algebraic techniques~\cite{FewsterV18}. While our interpretation of states differs slightly from FV's, our results represent an additional argument favouring the algebraic toolbox. In Sec.~\ref{s.conclusions}, we present our conclusions and briefly compare our results with those of Ref's~\cite{Polo-GomezEtAl21} and \cite{FewsterV18}. As we will show, this work builds a bridge between PGM's and FV's approaches, using their relative strengths as the foundation for a more user-friendly algebraic methodology for studying measurements in QFT.

\subsection{First issue: measuring collapsed states}
\label{s.issue1}
The first issue with contextual states relates to the possibility of observers not having the maximum available information about a system's state to perform measurements and obtain the correct results. In other words, if a state is only the summary of one's knowledge about a system, then observers having different information about the system will describe different states. Yet, all these observers should see the same physics, which generally contradicts different state assignments. This type of ignorance is more general than that one caused by the causal separation between an event and the observer speaking about it. Therefore, we distinguish two cases. On the one hand, an observer might be spacelike separated from the event they are considering; in this case, we say their ignorance has a causal (ontic) origin. On the other hand, an observer might be in the causal future of the event and yet be unaware of all its details, hence having a subjective type of ignorance that could possibly be removed by giving them enough information; in this case, we say their ignorance has a subjective (epistemic) origin. In the rest of this paragraph, we will always be concerned about this latter type of ignorance.

We illustrate the potential pitfalls of a naive definition of contextual states by the following example. Consider a non-relativistic two-level system in the state
\begin{equation}
    \ket{\psi}=\frac{1}{\sqrt{2}}\left(\ket{0}+\ket{1}\right)~,
    \label{e.context_two_level}
\end{equation}
and two observers Alice and Bob ($\mathcal{A}$ and $\mathcal{B}$) measuring the system in the computational basis $\{\ket{0},\ket{1}\}$.
The contextual states of Alice and Bob,  denoted by a lower index inside the ket (i.e. $ \ket{\psi_\mathcal{A}},\ket{\psi_\mathcal{B}}$), depend on their information 
about the state $\ket{\psi}$ and describe the probability of obtaining some measurement outcome over a collection of identically prepared systems according to a specific observer. Let us assume that initially both Alice and Bob know the system is prepared in the state \eqref{e.context_two_level}. Hence, their
contextual states are
\begin{equation}
\label{e.initial_contextual_state}
    \ket{\psi_\mathcal{A}}=\ket{\psi_\mathcal{B}}=\ket{\psi}~.
\end{equation} 
Next, Alice freely acts on it without Bob knowing. Specifically, 
suppose Alice performs a measurement and obtains an outcome 0, that she does not communicate to Bob\footnote{In order for Bob to collect enough data to compute expectation values, we can think Alice measures many systems, discarding those giving the outcome 1 and handing to Bob only those giving the outcome 0.}. Because of the collapse postulate, her contextual state
after this event is $\ket{\psi'_\mathcal{A}}=\ket{0}$, while Bob has not
gained any new information, so his contextual state is naively still $\ket{\psi'_\mathcal{B}}=\ket{\psi}$. In consequence,  Bob will obtain wrong statistics when trying to measure some observable $\hat{O}$, i.e. 
\begin{equation}
    \langle\hat{O}\rangle_\mathcal{B}=\bra{\psi}\hat{O}\ket{\psi}\neq\bra{0}\hat{O}\ket{0}= \langle\hat{O}\rangle_\mathcal{A}~.
\end{equation}

The resolution of the contradiction is that quantum states can be seen as contextual only if they include full knowledge of what happened to the system before it was handed to the observer. In other words, we need to accept \textit{"the existence of an “expert” whose probabilities we should strive to possess whenever we strive to be maximally rational"}~\cite{Fuchs02}, and that knows all details about the previous history of the system. We call this observer the \textit{Maximally Informed Observer} (MIO). Similarly, observers who are unaware of operations that happened in the past (or willingly ignore them) will describe a contextual quantum state which is not a good tool for foreseeing future outcomes. Therefore, contextual states come with the necessity of a MIO describing the best state possible, which can be regarded as somewhat \textit{more real} than the others.

Translating the above argument into relativistic QM adds the complication of causality. As we will see, PGM propose the state outside the light cone of a measurement can be either the one before the measurement or that describing a non-selective measurement, depending on the knowledge available to the observer. These exhaust all the possible situations one can encounter outside the measurement light cone.
However, using the same argument used in the non-relativistic example, an observer in the future of the measurement having limited access to the information about the outcome (and hence formally equivalent to one being outside the light cone, but for epistemic - not ontic - reasons) should use a state as seen by a MIO or get wrong results. As a result, the notion of MIO of the non-relativistic scenario acquires a spacetime dependence, for the maximal information available is always relative to some past spacetime region.

In Sec.~\ref{s.state_assignment}, we review PGM's argument for how a notion of spacetime dependence for states arises from the assumption of contextual states and, in Sec.~\ref{s.density states}, we propose a QFT state assignment to abandon contextuality entirely.

\begin{figure*}
     \centering
      \begin{subfigure}[b]{\columnwidth}
         \includegraphics[width=\textwidth]{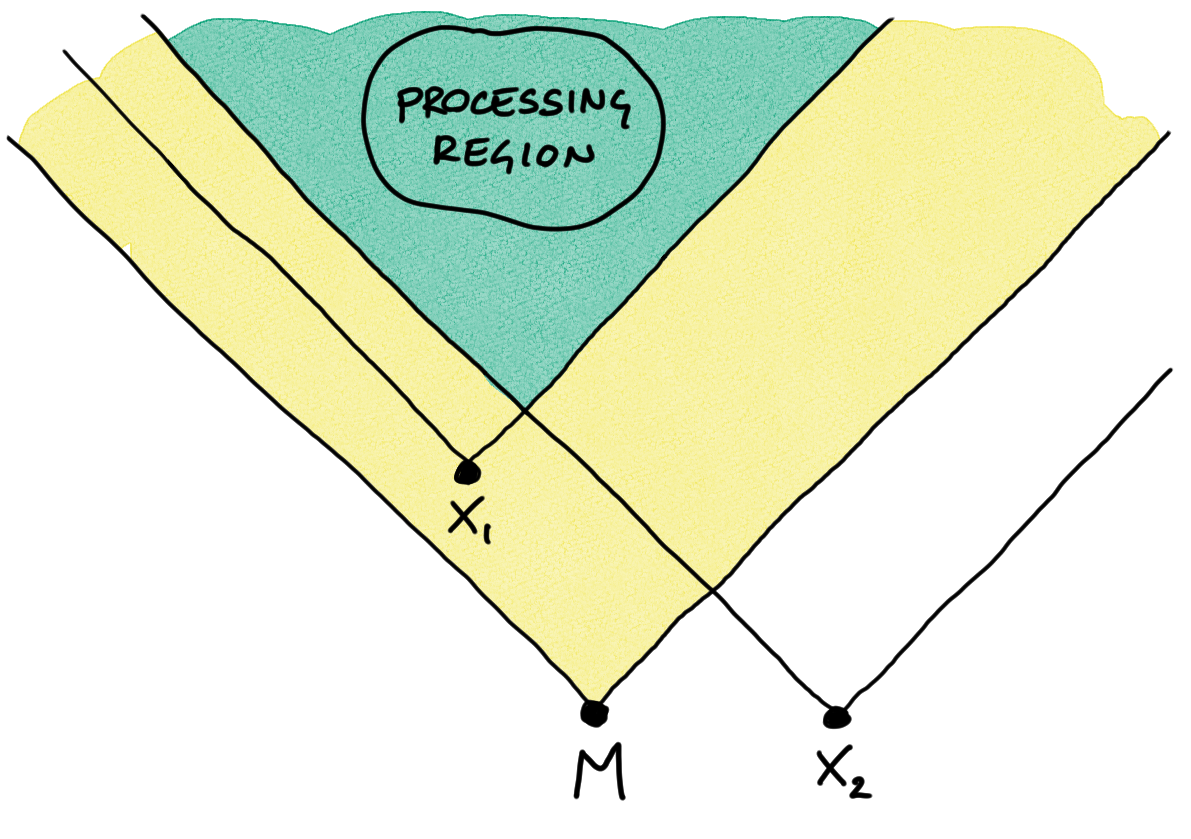}
         \caption{Pair of spacetime points for which PGM propose calculating $w(x_1, x_2)$ via {\bf P2}. The causal future of the measurement $M$ (yellow) contains the overlap of those of $x_1$ and $x_2$ (green). Consequently, the processing region for $w(x_1,x_2)$ is contained in the causal future of the measurement. As it is clear, information collected at both $x_1$ and $x_2$ is only available to MIOs knowing the measurement outcome, and the usage of {\bf P2} is justified.}
         \label{f.issue2a}
     \end{subfigure}
    \hfill
    \begin{subfigure}[b]{\columnwidth}
         \includegraphics[width=\textwidth]{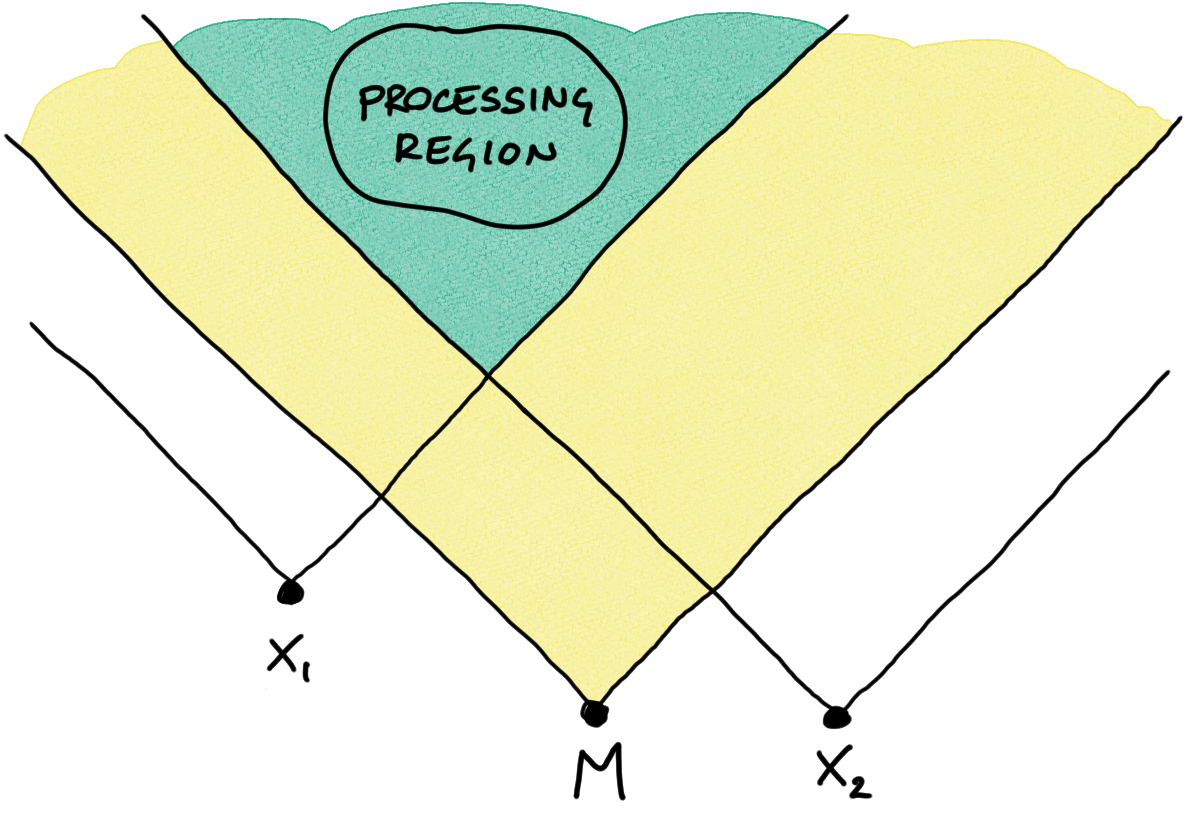}
         \caption{Pair of spacetime points for which calculating $w(x_1, x_2)$ via {\bf P1}-{\bf P2} is ill-defined. The causal future of the measurement $M$ (yellow) contains the overlap of the causal futures of $x_1$ and $x_2$ (green), and hence the processing region for $w(x_1,x_2)$. While {\bf P1} says to use the state of the MIO outside the yellow region, the logic leading to {\bf P2} says we should use the state of the MIO inside the yellow region, leading to a contradiction.}
         \label{f.issue2b}
     \end{subfigure}
        \caption{Localization of the processing region for the two-point function $w(x_1,x_2)$. On the left is an example of the setting for which {\bf P2} was introduced; on the right is one case for which using {\bf P1}-{\bf P2} leads to a contradiction.}
        \label{f.issue2}
\end{figure*}

\subsection{Second issue: two-point functions across and outside a measurement's future light cone}
\label{s.issue2}

The second issue arises from using PGM's contextual states to calculate $n$-point correlation functions. As evident, the proposed contextual nature of QFT states renders the usual rule for calculating $n$-point functions inadequate, as correlation functions not only depend on the field operators at the points they are inserted in, but also on the information available to the observers computing the given correlation function. To better formalise this idea, PGM introduce the notion of the \textit{processing region} of a $n$-point function. Given a set of spacetime points $\{x_1,\dots,x_n\}$ over which one wants to calculate the $n$-point function $w(x_1,\dots,x_n)$, the processing region is a portion of spacetime contained in the intersection of future lightcones of these points. Messages carrying information about the events labelled by $x_1,\dots,x_n$ can only reach observers in the processing region~\cite{Ruep21}. Therefore, computation of correlation functions for operators inserted in $x_1,\dots,x_n$ must be ``performed'' by MIOs in the processing region, able to use the appropriate state for the computation. Because of the explicit appearance of the MIO in this construction, PGM argued that moving from density operators to the proposed $n$-point functions can remove the intrinsic contextuality of the measurement scheme.

Let us view how the computation of $n$-point functions works in the case of a spacetime containing only one measurement $M$. In this case, PGM propose the following prescription\footnote{Notice that, while the original formulation of the prescription does not include the notion of MIO, it is essentially equivalent to the one presented herein.} to calculate $n$-point functions:
\begin{itemize}
    \item[{\bf P1}] if the points are all inside (all outside) the measurement's future light cone, use
    \begin{equation}
    \label{e.all_in}
        w(x_1,\dots,x_n)=\Tr[\rho\phi(x_1),\dots,\phi(x_n)]~,
    \end{equation}
    where $\rho$ is the contextual field's state relative to the MIO inside (outside) the future light cone of the measurement.
    \item[{\bf P2}] if the points are both inside and outside $M$'s future light cone, use
    \begin{equation}
    \label{e.some_in}
        w(x_1,\dots,x_n)=\Tr[\rho'\phi(x_1),\dots,\phi(x_n)]~.
    \end{equation}
    where $\rho'$ is the contextual field's state relative to the MIO inside the future light cone of the measurement.
\end{itemize}
The reason for using the state relative to the MIO inside the future light cone of the measurement in \textbf{P2} is that, for $n>1$, $n$-point functions are non-local quantities that must locally be computed. To do so, an observer willing to calculate a $n$-point function should wait until the necessary (locally collected) information reaches some processing region which, in \textbf{P2}, is always included in the future light cone of the measurement; hence, the state relative to the MIO of that region must be used in this case (see Fig.~\ref{f.issue2a}). 


However, the above prescription is inconsistent for $n$-point functions evaluated outside the light cone. To show this, let us consider $w(x_1,x_2)$ with $x_1$ and $x_2$ outside the future of the measurement and such that the overlap of their causal futures is fully contained in the future of the measurement, as in Fig.~\ref{f.issue2b}. As both points lay outside the future of the measurement, PGM's prescription tells us to use \textbf{P1} and calculate $w(x,y)$ by the state of the MIO of the observer laying outside the future of the measurement. However, the processing region of these points is contained in the future light cone of the measurement. Hence, any observer computing the two-point function has the knowledge given by the selective update and should use {\bf P2} instead. Therefore, \textbf{P1} and  \textbf{P2} give conflicting prescriptions for this scenario.

Another example in which the above prescription is ill-defined is when the causal future of the two points is only partially contained within the future of the measurement, as in Fig.~\ref{f.issue2c}. In this case, the processing region can be contained inside and/or outside the causal future of $M$, hence leading to an additional ill-definiteness of \textbf{P1}-\textbf{P2} related to the need to say where the observer is while computing the correlation function. In other words, $n$-point functions also acquire an explicit observer dependency and are thus manifestly as subjective as the density operators we started from. Hence, this example shows that PGM's prescription is insufficient to fully remove contextuality, and that a shift in the mathematical language may be needed.

In Section~\ref{s.algebraic states}, we propose employing algebraic tools for addressing the first issue, and we discuss their effectiveness in resolving the inconsistencies related to $n$-point functions.

\begin{figure}
         \includegraphics[width=\columnwidth]{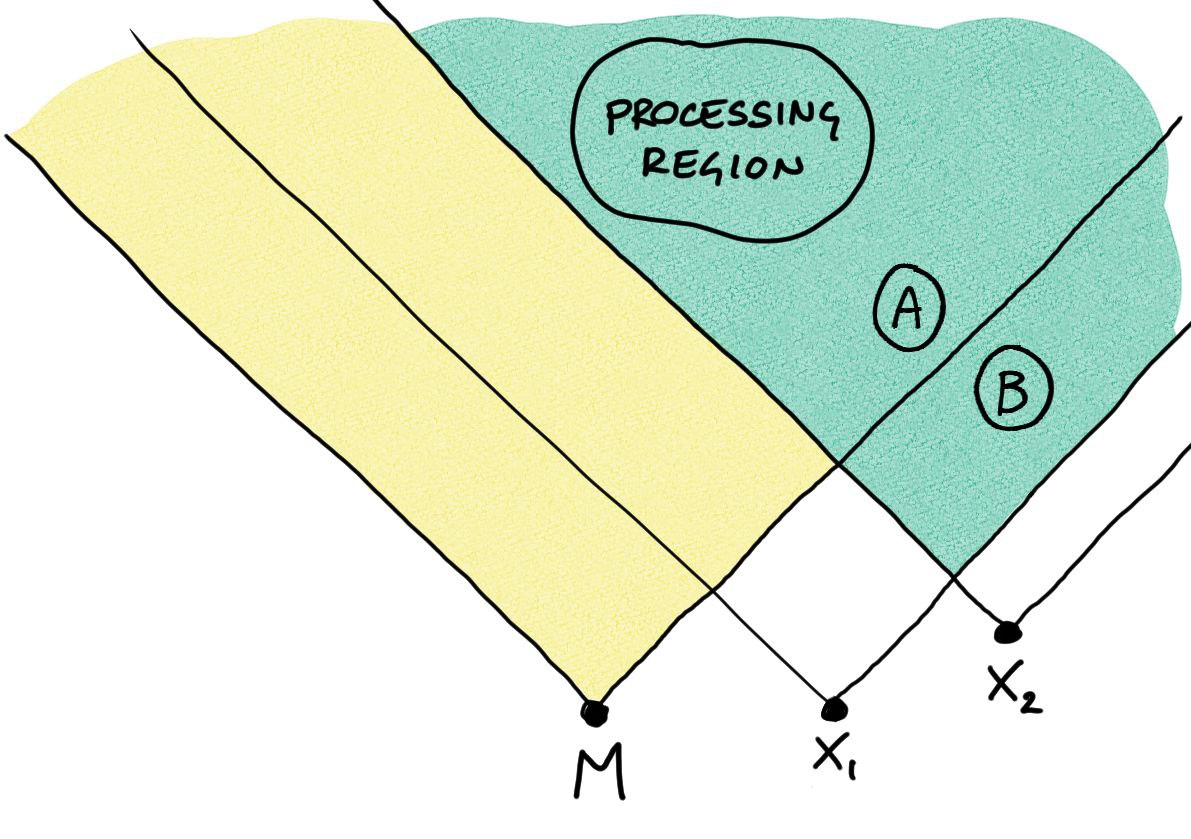}
         \caption{A further selection of points $x_1$ and $x_2$ for which calculating $w(x_1, x_2)$ via \textbf{P1}-{\bf P2} is problematic. The portion of spacetime amenable to containing a processing region (green) overlaps the causal future of the measurement $M$ (yellow) without being contained in it. Consequently, an observer in subregion $A$ would compute $w(x_1, x_2)$ by \textbf{P2}, while one in subregion $B$ would use \textbf{P1}, hence leading to another ill-definiteness of the prescription.}
         \label{f.issue2c}
\end{figure}

\subsection*{Notation and Conventions}
Throughout the article, we use the following notation and conventions. We denote spacetime points with lowercase letters (e.g. $x$, $y$, etc.) and their spatial part by bold lowercase letters (e.g. $\mathbf{x}$, $\mathbf{y}$, etc.). Sets of spacetime points are called regions and denoted by uppercase letters (e.g. $A$, $B$, etc.). As it will be discussed in Sec.~\ref{s.measurements_events}, measurements are also spacetime regions and hence are denoted by uppercase letters. The spacetime is Hausdorff, time-orientable and has metric with signature $(-,+,\dots)$. Moreover, we need to assume the dimension of the spacetime to be larger than two for technical reasons which will become clear in Sec.~\ref{s.second_solution}. The causal future of a spacetime region $A$, denoted by $\mathcal{J}^+(A)$, is the set of all points which can be reached from $A$ by future-directed non-spacelike curves, and the domain of future dependence of $A$, denoted by $\mathcal{D}^+(A)$, is the set of all points $x$ such that every past-inextendible non-spacelike curve through $x$ intersects $A$. The causal past $\mathcal{J}^-(\cdot)$ and domain of past dependence $\mathcal{D}^-(\cdot)$ are defined similarly~\cite{HawkingEllis73}. From these, one builds the regions $\mathcal{J}(A)=\mathcal{J}^+(A)\cup\mathcal{J}^-(A)$, $\mathcal{D}(A)=\mathcal{D}^+(A)\cup\mathcal{D}^-(A)$, and $\mathcal{S}_A=\mathcal{M}\setminus\mathcal{J}(A)$. Finally, observers are denoted by calligraphic capital letters (e.g. $\mathcal{A}$ for Alice, $\mathcal{B}$ for Bob, etc.).

\section{Refined detector-based measurements in QFT}
\label{s.setup}

The Unruh-DeWitt model for particle detectors was introduced as a tool for studying the notion of particles in QFT in curved and non-inertial spacetimes~\cite{Unruh76,DeWitt80,BirrellD82, Wald95, CrispinoEtAl08}. For the purpose of this article, it is sufficient to consider the simple setting of a bipartite system composed of a two-level Unruh-DeWitt detector $M$ with non-degenerate free Hamiltonian
\begin{equation}
    \hat{H}_{M}=\omega\ket{1}\bra{1}~,
\end{equation}
and a real massive scalar quantum field $\phi$ with Hamiltonian $\hat{H}_\phi$. In addition to the global free Hamiltonian
\begin{equation}
    \hat{H}=\hat{H}_{M}+\hat{H}_\phi
    \label{H_free}
\end{equation}
the detector and the field interact via the (detector's proper) time-dependent Hamiltonian
\begin{equation}
    \hat{H}_{\mathrm{int}}(\tau)=\chi(\tau)\hat{m}(\tau)\otimes\hat{\phi}_f(\tau)~,
    \label{H_int}
\end{equation} 
where $\chi(\tau)$ is a smooth switching function describing the interaction's duration and intensity,
\begin{equation}
\hat{m}(\tau)=e^{\mathrm{i}\omega \tau}\ket{1}\bra{0}+e^{-\mathrm{i}\omega \tau}\ket{0}\bra{1}
\end{equation}
is the detector's monopole momentum operator in the interaction picture, and $\hat{\phi}_f(\tau)$ is a spatially
smeared field operator describing the detector's spatial profile~\cite{Schlicht04, LoukoS06, Satz07} along the  detector's world line
\begin{equation}
    x(\tau)=(t(\tau),\mathbf{x}(\tau))~.
\end{equation}
Finally, $\hat{H}_{\mathrm{int}}(\tau)$ is often multiplied by a small constant $\lambda$ to make the interaction weak and allow the usage of perturbation theory. In PGM's framework, these detectors are used as a probe to test the field; in the rest of this section, we discuss how this is made, refining some details about the model along the way.

\subsection{Non-relativistic measurements}
In order to frame our later discussion, it is crucial to have a robust description of non-relativistic measurements\footnote{In this work, we call non-relativistic all those systems for which the usual postulates of QM (including the measurement one) hold valid for all practical purposes.} and their reading as sets of spacetime events. In non-relativistic QM, one can introduce a measurement by the so-called positive operator-valued measure (POVM): a map $O$ from the family of parts $\mathcal{M}$ of the set of possible outcomes of the measurement $R=\{E_0,E_1\dots\}$ to the set of Hermitian operators acting on the Hilbert space of the measured system, i.e.
\begin{equation}
    O:\mathcal{M}\longrightarrow op(\mathcal{H})~,
\end{equation}
such that
\begin{align}
    &\hat{O}(E_l)\equiv\hat{E}_l\geq 0\label{POVM1}\\
    &\hat{O}(R)=\mathbb{I}\label{POVM2}\\
    &\hat{O}(\cup_l E_l)=\sum_l \hat{O}(E_l)=\sum_l \hat{E}_l\label{POVM3}
\end{align}
(see e.g. Ref.~\cite{HeinosaariZ11}). In particular, since the measurement operators are positive, we can always find a set of operators $\hat{M}_i$ such that
\begin{equation}
    \hat{E}_i=\hat{M}^\dagger_i\hat{M}_i~,
\end{equation}
and hence rewrite Eq.~\eqref{POVM2} as
\begin{equation}
    \sum_i \hat{M}^\dagger_i\hat{M}_i=\mathbb{I}~.
    \label{e.MeasIdentity}
\end{equation}

Although the above framework provides a mathematical description of measurements, it does not encompass all the physical aspects involved in actual measurements performed in labs. Specifically, a quantum measurement involves a collapse phenomenon that extends beyond POVMs alone. Nevertheless, one can define a minimal description of measurements via POVMs as follows. For any element $E_i\in R$, the number 
\begin{equation}
    p^{O}_{\rho}(E_i)=\Tr[\rho \hat{O}(E_i)]
\end{equation}
is the probability that an experimental observation associated with the measurement $O$ of the system in the state $\rho$ gives the result $E_i$. If the outcome $E_i$ is realized, the state of the system makes the (non-unitary) jump
\begin{equation}
    \rho\longmapsto\rho_i=\frac{\hat{M}_i\rho\hat{M}_i^\dagger}{\Tr[\hat{M}_i\rho\hat{M}_i^\dagger]}~.
    \label{e.collapse}
\end{equation}
The update of the system's state via Eq.~\eqref{e.collapse} is usually called measurement-induced state collapse, or L\"uders rule.

\subsection{Detector-based measurements as sets of events}
\label{s.measurements_events}
PGM's framework utilises non-relativistic measurements to establish a measurement procedure in the context of QFT. This is achieved through non-relativistic systems, such as Unruh-DeWitt detectors, as mediators to measure quantum fields. Specifically, one lets a  non-relativistic system interact with a quantum field and measures the former to extract information about the latter. Here, we refine PGM's measurements on quantum fields by considering all relevant parts of non-relativistic measurements in the description, including the (unspecified) process that generates the measurement output. Therefore, a detector-based measurement on a quantum field is composed of three processes:
\begin{enumerate}
    \item the detector and the field are entangled via some unitary interaction. This allows for later extraction of information about the field by measurements performed on the detector via the von Neumann scheme~\cite{vonNeumann55, Ozawa84}. The spacetime region containing (at least) all events comprising this interaction is called $M_c$. Because we require the switching function $\chi$ and the spatial smearing $f$ to have compact support, this region is closed.
    \item the POVM performed on the detector and the related output production, which is the process making the (non-unitary) transition \eqref{e.collapse} to occur. Different interpretations of QM disagree on the specific microscopic realization of the output production~\cite{Maudlin95}; for the sake of generality, here we only assume there is some kind of unspecified process realizing \eqref{e.collapse}, which can be represented by a collection of events included in a spacetime region $M_o$. Because this region is contained in $\mathcal{J}(M_c)$, we can take it to be closed.
    \item a delay, necessary for making a large enough region of spacetime (typically having the size around that of the detector) aware of the obtained result. The delay is defined as a set $M_d$ such that every $p\in M_d$ satisfying $\mathcal{J}^+(p)\cap M_d=\{p\}$ is contained in the future of the future endpoint of lightlike curves on the boundary of $\mathcal{D}^+(M_o)$. Because the future endpoint of lightlike curves on the boundary of $\mathcal{D}^+(M_o)$ is a closed set, we can take $M_d$ to be closed.
\end{enumerate}
Notice that the relation $M_o\subset \mathcal{J}^+(M_c)$ must be satisfied. For the sake of simplicity, we assume that the three spacetime subsets are connected, and any measurement $M$ is simply the union $M=M_c\cup M_o\cup M_d$. Being the union of closed sets, $M$ is closed. Furthermore, we call $\mathcal{P}^+_M$ the region in which the classical information about the measurement's outcome is available, which corresponds to the causal future of all points $p\in M$ satisfying $\mathcal{J}^+(p)\cap M=\{p\}$. In analogy, we also define $\mathcal{P}^-_M$ as the causal past of all points $p\in M$ satisfying $\mathcal{J}^-(p)\cap M=\{p\}$. An example of a measurement's spacetime structure is represented in Fig.~\ref{f.regions}.

\begin{figure}
    \centering
    \includegraphics[width=0.8\columnwidth]{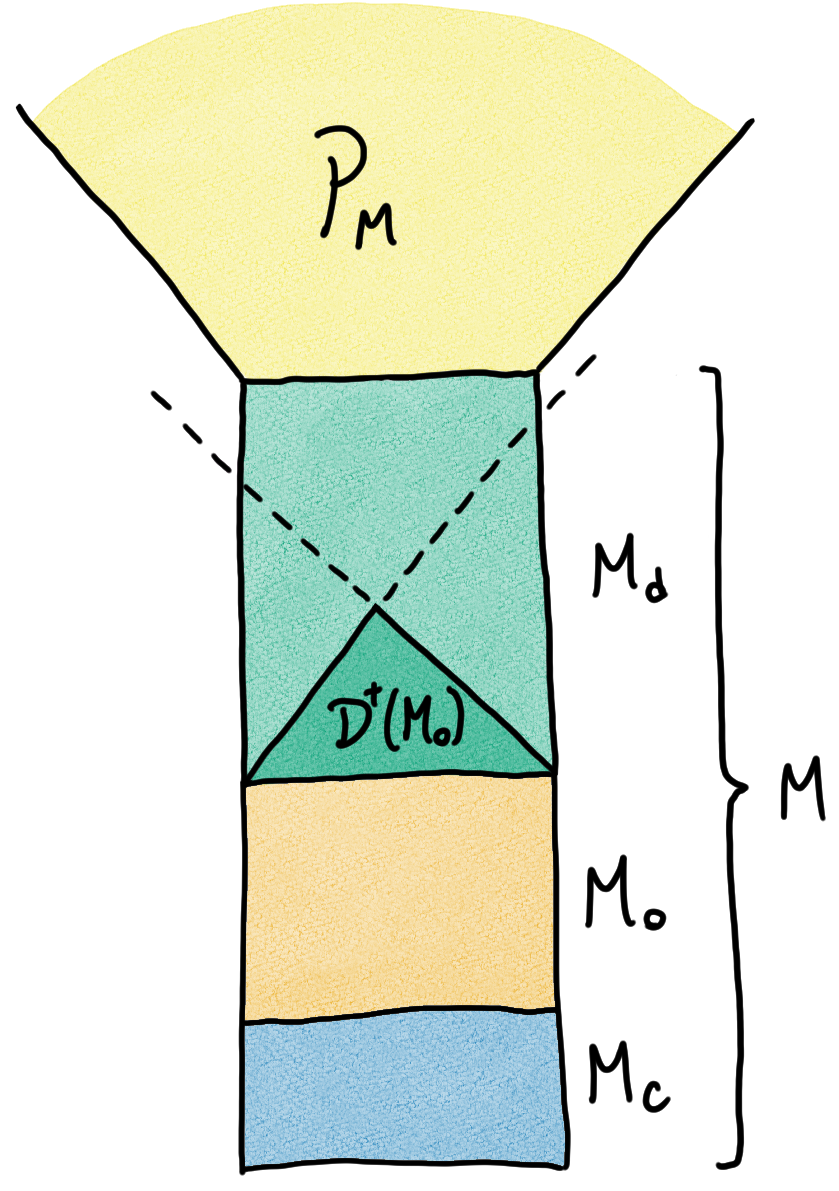}
    \caption{Spacetime structure of a measurement. The measurement is composed of the detector-field coupling $M_c$ (blue), the output production (orange) and a delay region $M_d$ (green) built in such a way that all points whose future contains no measurement points have access to the measurement output. To obtain this, it is sufficient that these are in the future of the "latest" point of $\mathcal{D}^+(M_o)$ (dark green), whose future light cone is represented by dotted lines. Finally, the region $\mathcal{P}^+_M$ (yellow) is a collection of events where the measurement output is accessible.}
    \label{f.regions}
\end{figure}

\section{Two detectors: a case study}
\label{s.state_assignment}

Following PGM's analysis in Ref.~\cite{Polo-GomezEtAl21}, let us now consider the case where two detectors $A$, $B$ interact with the same massive real scalar field $\phi$. Given a spacelike foliation of the spacetime $\Sigma_t$ parametrized by $t$, the interaction Hamiltonian of the two detectors with the field is taken to be
\begin{equation}
    \hat{H}_{int}(t)=\hat{H}_A(\tau_A)+\hat{H}_B(\tau_B)~,
\end{equation}
where the interaction Hamiltonians are those of Eq.~\eqref{H_int} and the relation $\tau_{A,B}(t)$ is provided by $\Sigma_t$. In the following, we identify a detector with its measurement region and consider:
\begin{itemize}
    \item the future light cone $\mathcal{J}^+(I)$ of the detector $I$, for $I=A,B$;
    \item the future light cone of the events corresponding to the measurement output production and read-out $\mathcal{P}^+_I$, for $I=A,B$.
\end{itemize}

Following Ref.~\cite{Polo-GomezEtAl21}, we study the field as tested by a third detector $C$ when the detectors $A$ and $B$ are spacelike to each other, and neither $A$ nor $B$ are in the future of $C$. Three scenarios are then defined by the causal relation between $A$ and $B$ and the third detector. In particular, $C$ can be:
\begin{enumerate}
    \item spacelike separated from both $A$ and $B$, i.e. $C\subset\mathcal{S}_{AB}$: an observer measuring $C$ after its interaction with the field would have no information about the outcomes of measurements performed on $A$ or $B$.
    \item in $\mathcal{P}^+_A\setminus\mathcal{P}^+_B$ (or $\mathcal{P}^+_B\setminus\mathcal{P}^+_A$): an observer can access the measurement outcome of $A$ ($B$) but not that of $B$ ($A$).
    \item in $\mathcal{P}^+_A\cap\mathcal{P}^+_B$: an observer measuring $C$ can access both the outcomes obtained by $A$ and $B$.
\end{enumerate}
These three possibilities are summarized in Fig.~\ref{fig.cases}, where the relevant spacetime regions are highlighted by different colours. 

Working under the assumption that no interaction happened before the measurements, the composite system made of the detectors and the field can be taken to initially be in the pure state
\begin{equation}
\label{e.init}
    \ket{\psi}= \ket{\phi_0}\otimes\ket{\alpha}\otimes\ket{\beta}~.
\end{equation}
Hence, the initial state of the field in the density operator formalism is
\begin{equation}
    \rho_{\phi_0}=\ket{\phi_0}\bra{\phi_0}~.
\end{equation}
Notice that while this state depends on the specific slicing used to define the time evolution of spacelike sub-manifolds, there is always a way to select a Cauchy slice in the past of the interactions (i.e. that it does not intersect the detector regions) for which the state of the field is a projector. For later convenience, we also define the initial detectors' density operators
\begin{equation}
    \sigma_\kappa=\ket{\kappa}\bra{\kappa}
\end{equation}
for $\kappa=\alpha,\beta$; therefore, the state in Eq.~\eqref{e.init} can also be expressed as
\begin{equation}
    \rho_0=\rho_{\phi_0}\otimes\sigma_\alpha\otimes\sigma_\beta~.
\end{equation}
Moreover, since $A$ and $B$ are spacelike separated, the operators in $\hat{H}_A(\tau_A),\hat{H}_B(\tau_B)$ commute, hence the 
the time evolution operator $U(t;t_0)$ in the interaction picture satisfies
\begin{equation}
    U(t;t_0)=U_A(t;t_0)U_B(t;t_0)~,
\end{equation}
with $[U_A(t;t_0),U_B(t;t_0)]=0$, regardless of the slicing considered. Furthermore, it is worth noting that different slicing choices give different field and detector states at the same spatial location: while surprising, it will be clear from the later discussion that this does not affect physical results, and, in our proposal, all these states are related by an equivalence relation making the slicing choice irrelevant.

Before analysing the above three cases (already presented in a slightly different fashion in Ref.~\cite{Polo-GomezEtAl21}), let us briefly mention that these do not exhaust all possible regions where one can find $C$. For example, the spacetime strips $\mathcal{J}^{+}(I)\setminus\mathcal{P}^+(I)$, $I=A,B$, are not covered by the above discussion. This is due to the difficulty of modelling contemporary interactions. The case in which $C$ is in these spacetime regions will not be considered for the rest of the article.

We now study how to assign a state for a scalar quantum field $\phi$ in the case of the detector $C$ being measured when in the regions
\begin{enumerate}
    \item $C\subset\mathcal{S}_{AB}$
    \item $C\subset\mathcal{P}^+_A\setminus\mathcal{P}^+_B$ and $C\subset\mathcal{P}^+_B\setminus\mathcal{P}^+_A$
    \item $C\subset\mathcal{P}^+_A\cap\mathcal{P}^+_B$
\end{enumerate}
The only region having a clear causal relation with $A$ and $B$ left out of our analysis are $\mathcal{P}^-_A$ and $\mathcal{P}^-_B$, for which the discussion is trivial. We remark that the discussion of the next three sections (Sec.s~\ref{s.first_case} to \ref{s.third_case}) closely follows that reported in Sec.~VIII of Ref.~\cite{Polo-GomezEtAl21}, the only differences being the notation used, some of the formal tools employed (see App.s~\ref{a.ProofFactoringMeansIdentity} and \ref{a.proof5689}), and the interpretation we give at the same results.

\begin{figure}[t]
\centering
    \includegraphics[width=\columnwidth]{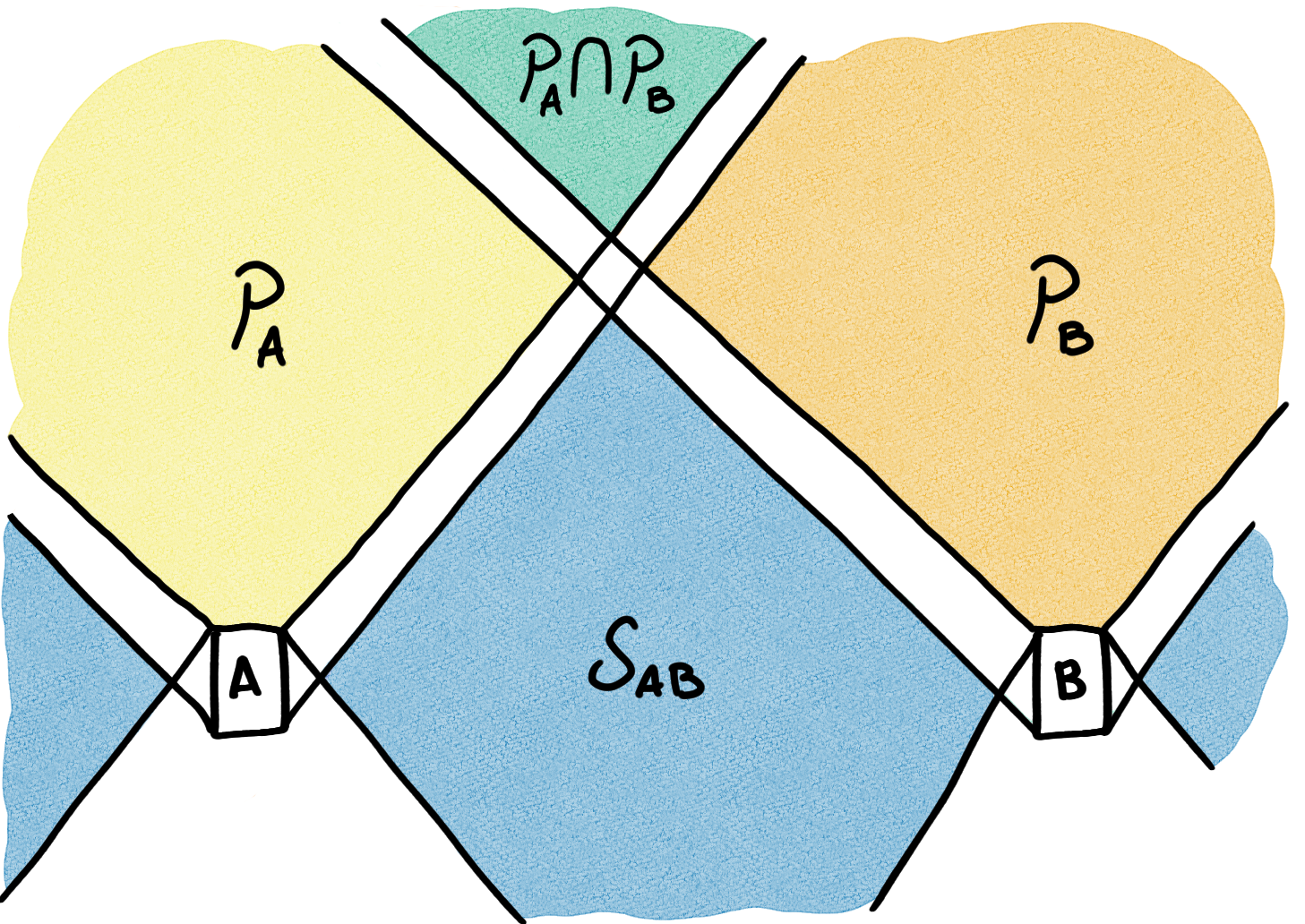}
    \caption{Causal diagram for two space-like separated detectors. The blue region, called $\mathcal{S}_{AB}$, is the region space-like separated from the detectors $A$ and $B$. The yellow (orange) region, called $\mathcal{P}^+_A$ ($\mathcal{P}^+_B$), is the future light cone of a set of points containing the events where the output of a measurement performed on $A$ ($B$) is known. The green region, called $\mathcal{P}^+_A\cap\mathcal{P}^+_B$, is the set of points that are in the future light-cone of the output production of both $A$ and $B$.}
    \label{fig.cases}
\end{figure}

\subsection{\texorpdfstring{$C\subset\mathcal{S}_{AB}$}{C in S[AB]}}
\label{s.first_case}
First, we consider the case when $C$ is spacelike to $A$ and $B$, i.e. $C\subset \mathcal{S}_{AB}$. In this case, the possible states of knowledge of $\mathcal{O}_C$ about the measurements performed on $A$ and $B$ are:
\begin{enumerate}[series=poss1, label=\textbf{\arabic*)}]
    \item $\mathcal{O}_C$ does not know $\mathcal{O}_A$ and $\mathcal{O}_B$ performed measurements \\$\Rightarrow$ they describe the field's state as $\rho_{\bf 1}\equiv\rho_{\mathcal{S}_{AB}}=\rho_{\phi_0}$;
    \item $\mathcal{O}_C$ only knows $\mathcal{O}_A$ performed a measurement (but not the outcome) \\$\Rightarrow$ they describe the field's state as $\rho_{\bf 2}\equiv\rho^A_{\mathcal{S}_{AB}}$;
    \item $\mathcal{O}_C$ only knows $\mathcal{O}_B$ performed a measurement (but not the outcome) \\$\Rightarrow$ they describe the field's state as $\rho_{\bf 3}\equiv\rho^B_{\mathcal{S}_{AB}}$;
    \item $\mathcal{O}_C$ knows that both $\mathcal{O}_A$ and $\mathcal{O}_B$ performed measurements (but not the outcomes) \\$\Rightarrow$ they describe the field's state as $\rho_{\bf 4}\equiv\rho^{AB}_{\mathcal{S}_{AB}}$;
\end{enumerate}
The states related to the above points \textbf{2} and \textbf{3} are obtained by the partial trace over $\mathcal{H}_A$ and $\mathcal{H}_B$ of the state of as it is described by $\mathcal{O}_C$. For example, in the case \textbf{2} the state before the trace is
\begin{equation}
\left(\sum_a\mathbb{I}\otimes\hat{M}_a(\hat{U}_A\rho_{\phi_0}\otimes\rho_{\alpha}\hat{U}_A^\dagger)\mathbb{I}\otimes\hat{M}^\dagger_a\right)\otimes\rho_{\beta}\otimes\rho_{\gamma}
\label{e.state_sum_S_AB}
\end{equation}
which, after the partial trace, gives
\begin{equation}
    \rho_{\bf 2}= \rho^A_{\mathcal{S}_{AB}}=\Tr_A[\hat{U}_A(\rho_{\phi_0}\otimes\rho_{\alpha})\hat{U}_A^\dagger]~;
\end{equation}
similarly, in case \textbf{3} one gets
\begin{equation} 
    \rho_{\bf 3}= \rho^B_{\mathcal{S}_{AB}} =\Tr_B[\hat{U}_B(\rho_{\phi_0}\otimes\rho_{\beta})\hat{U}_B^\dagger]~.
\end{equation}
Finally, the density operator $\rho_{\bf 4}$ is obtained by the same procedure, where now the states of both $A$ and $B$ are those obtained by a non-selective update rule, hence giving
\begin{equation}
    \rho_{\bf 4}=\Tr_{AB}[\hat{U}_A\hat{U}_B(\rho_{\phi_0}\otimes\rho_{\alpha}\otimes\sigma_\beta)\hat{U}_A^\dagger\hat{U}^\dagger_B]~.
\end{equation}

For the picture to be consistent, it must be that all locally supported field observables $\mathcal{O}_C$ one can measure have the same expectation value when evaluated via the density operators \textbf{1-4}, i.e. 
\begin{equation}
    \Tr[\rho_{\bf 1}\hat{L}]=\Tr[\rho_{\bf 2}\hat{L}]=\Tr[\rho_{\bf 3}\hat{L}]=\Tr[\rho_{\bf 4}\hat{L}]
    \label{e.LinC}
\end{equation}
when $\hat{L}$ is an observable supported in $\mathcal{S}_{AB}$. Specifically, we call $\mathfrak{O}(R)$ the set of all observables whose support is causally convex\footnote{The requirement for the region to be causally convex comes from the fact that we require information and fields to be local, as customary in the literature on local quantum physics~\cite{Haag92} (see also Ref.~\cite{FewsterV18}).} and contained in the region $R$. The validity of the above chain of equalities was obtained in Ref.~\cite{Polo-GomezEtAl21} by using the cyclic property of the trace and the fact that, since $A$, $B$ and $\mathcal{S}_{AB}$ are space-like separated, $[\hat{U}_A,\hat{U}_B]=0$ and
\begin{equation}
    [\hat{L},\hat{U}_A]=[\hat{L},\hat{U}_B]=0~.
    \label{e.LUcommutator}
\end{equation}

Hence, an observer that is spacelike to both $A$ and $B$ can use any density operator among \textbf{1}-\textbf{4} and get the right expectation values. More generally, any density operator in the equivalence class
\begin{equation}
    \varrho_{\mathcal{S}_{AB}}=\{\rho ~|~ \Tr[\rho\hat{L}]=\Tr[\rho_{\phi_0}\hat{L}],~\forall\hat{L}\in\mathfrak{O}(\mathcal{S}_{AB})\}
    \label{e.EquivClassC}
\end{equation}
gives the correct physical results and can be used as the field's density operator.

As it is clear from the above discussion about the possible states of knowledge of $\mathcal{O}_C$, the elements of $\varrho_{\mathcal{S}_{AB}}$ are contextual to the observers in $\mathcal{S}_{AB}$. In particular, assigning a field state $\rho\in\varrho_{\mathcal{S}_{AB}}$ is not effective when considering single shot experiments (for, in this case, one cannot average over the ignorance of $\mathcal{O}_C$, i.e. over the many realisations appearing in the sum \eqref{e.state_sum_S_AB}) or correlations between $\mathcal{S}_{AB}$ and any region within the causal futures of the measurements (because these depend on the state of the field outside the region of validity of the state assignment). In other words, the above field states are only valid as far as it concerns the computation of expectation values of many \textit{i.i.d.} experiments localised in $\mathcal{S}_{AB}$. As we will see in Sec.~\ref{s.Bell_states_collapse}, this is similar to what happens when assigning the state $\rho_A={\rm Tr}[\rho_{AB}]=\frac{1}{2}\mathbb{I}$ to a particle $A$ composing a Bell pair with another particle $B$: the assignment can only provide reasonable statistics for local measurements performed on the partition defined by $A$ of many replicas of the Bell pair $AB$ while, on single-shot experiments, the state of $A$ should account for the perfect correlations between the outcomes of spacelike measurement performed on $A$ and $B$. This remark will become important in later sections when discussing the possibility of finding a state assignment to provide $n$-point functions evaluated on points laying across measurements' lightcones.

\subsection{\texorpdfstring{$C\subset\mathcal{P}^+_A\setminus\mathcal{P}^+_B$ and $C\subset\mathcal{P}^+_B\setminus\mathcal{P}^+_A$}{C in P[A]/P[B] and C in P[B]/P[A]}}
\label{s.second_case}
Next, let us consider the case in which the detector $C$ is space-like to $B$ and in the causal future of $\mathcal{P}^+_A$, i.e. when $C\subset\mathcal{P}^+_A\setminus\mathcal{P}^+_B$. In addition to the possible contextual field's states of knowledge \textbf{1}-\textbf{4} above, an observer in this region can also describe the field state in accordance with the following:
\begin{enumerate}[resume=poss1, label=\textbf{\arabic*)}]
    \item $\mathcal{O}_C$ only knows $\mathcal{O}_A$ performed a measurement and the outcome $a$\\$\Rightarrow$ they describe the field's state as $\rho_{\bf 5}= \rho^{A(a)}_{\mathcal{P}^+_A}$;
    \item $\mathcal{O}_C$ knows $\mathcal{O}_A$ performed a measurement and the outcome $a$, and knows $\mathcal{O}_B$ performed a measurement but not the outcome\\$\Rightarrow$ they describe the field's state as $\rho_{\bf 6}=\rho^{A(a)B}_{\mathcal{P}^+_A}$;
\end{enumerate}
The above states are obtained by the L\"uders rule as
\begin{equation}
    \rho^{A(a)}_{\mathcal{P}^+_A}=\frac{\Tr_A[\hat{M}_a\hat{U}_A(\rho_{\phi_0}\otimes\sigma_\alpha)\hat{U}^\dagger_A\hat{M}^\dagger_a]}{\Tr[\hat{M}_a\hat{U}_A(\rho_{\phi_0}\otimes\sigma_\alpha)\hat{U}^\dagger_A\hat{M}^\dagger_a]}
    \label{e.contextAnoB}
\end{equation}
and
\begin{equation}
    \rho^{A(a)B}_{\mathcal{P}^+_A}
    =\frac{\Tr_{AB}[\hat{M}_a\hat{U}_B\hat{U}_A(\rho_{\phi_0}\otimes\sigma_\alpha\otimes\sigma_\beta)\hat{U}^\dagger_A\hat{U}^\dagger_B\hat{M}^\dagger_a]}{\Tr[\hat{M}_a\hat{U}_B\hat{U}_A(\rho_{\phi_0}\otimes\sigma_\alpha\otimes\sigma_\beta)\hat{U}^\dagger_A\hat{U}^\dagger_B\hat{M}^\dagger_a]}
    \label{e.contextAsomeB}
\end{equation}
where 
we used the cyclic property of the trace and Eq.~\eqref{e.MeasIdentity}. 

Clearly enough, one could consider the case where $C$ lies in $\mathcal{P}^+_B\setminus\mathcal{P}^+_A$, and get similar results. In this case, the possible contextual field's state is one of the \textbf{1}-\textbf{4} above or one of the following:
\begin{enumerate}[resume=poss1, label=\textbf{\arabic*)}]
    \item $\mathcal{O}_C$ only knows $\mathcal{O}_B$ performed a measurement and the outcome $b$\\$\Rightarrow$ they describe the field's state as $\rho_{\bf 7}=\rho^{B(b)}_{\mathcal{P}^+_B}$;
    \item $\mathcal{O}_C$ knows $\mathcal{O}_B$ performed a measurement and the outcome $b$, and knows $\mathcal{O}_A$ performed a measurement but not the outcome\\$\Rightarrow$ they describe the field's state as $\rho_{\bf 8}=\rho^{AB(b)}_{\mathcal{P}^+_B}$;
\end{enumerate}
One obtains the states \textbf{7} and \textbf{8} in the same way as (\ref{e.contextAnoB}-\ref{e.contextAsomeB}) but with different upper and lower indices. The compatibility between \textbf{5} and \textbf{6} (\textbf{7} and \textbf{8}) was proven in Ref.~\cite{Polo-GomezEtAl21} by the same tools used for proving the equivalence of \textbf{1}-\textbf{4}. 

However, \textbf{5}-\textbf{6} are not compatible with \textbf{1}-\textbf{4} (while this could be proven also for \textbf{7}-\textbf{8}, here we limit our analysis to the former for the sake of simplicity). To prove this statement, we consider some field operator $\hat{L}\in\mathfrak{O}(\mathcal{P}^+_A\setminus\mathcal{P}^+_B)$ and confront its expectation value as observed by the observer $\mathcal{O}_A$ (having maximal information available, hence using the density operator $\rho_{\bf 6}$ or, equivalently, $\rho_{\bf 5}$) and $\mathcal{O}_C$, using any density operator in the equivalence class \eqref{e.EquivClassC}. As claimed above, when using $\rho\in\varrho_{\mathcal{S}_{AB}}$ the expectation value of an operator $\hat{L}$ is
\begin{equation}
    \langle \hat{L}\rangle_{\rho}=\Tr[\rho_{\phi_0}\hat{L}]~,
\end{equation}
which can be expressed equivalently by means of $\rho_{\bf 2}$ as
\begin{equation}
    \langle \hat{L}\rangle_{\bf 2}=\Tr[\hat{U}_A(\rho_{\phi_0}\otimes\sigma_\alpha)\hat{U}^\dagger_A\hat{L}]~.
\end{equation}
On the contrary, using $\rho_{\bf 5}$ gives
\begin{equation}
    \langle \hat{L}\rangle_{\bf 5}= \frac{\Tr[\hat{M}_a\hat{U}_A(\rho_{\phi_0}\otimes\sigma_\alpha)\hat{U}^\dagger_A\hat{M}^\dagger_a\hat{L}]}{\Tr[\hat{M}_a\hat{U}_A(\rho_{\phi_0}\otimes\sigma_\alpha)\hat{U}^\dagger_A\hat{M}^\dagger_a]}~.
\end{equation}
By calling $\rho'=\hat{U}_A\rho_{\phi_0}\otimes\sigma_\alpha\hat{U}^\dagger_A$ and $\hat{M}_a\hat{M}_a^\dagger=\hat{\Gamma}$, we get that 
\begin{equation}
    \langle \hat{L}\rangle_{\bf 2}= \langle \hat{L}\rangle_{\bf 5}~\Leftrightarrow~\Tr[\rho'\hat{L}\hat{\Gamma}]=\Tr[\rho'\hat{L}]\times \Tr[\rho'\hat{\Gamma}]~,
\end{equation}
for all $\hat{L}$, which is only satisfied when $\hat{\Gamma}=\mathbb{I}$ (see App.~\ref{a.ProofFactoringMeansIdentity}), and hence
\begin{equation}
    \langle \hat{L}\rangle_{\bf 2}= \langle \hat{L}\rangle_{\bf 5}~\Leftrightarrow~\hat{M}^\dagger_a\hat{M}_a\propto\mathbb{I}~.
    \label{e.FactoringMeansIdentity}
\end{equation}
However, in order to describe a non-trivial measure, the measurement operator must satisfy $\hat{M}^\dagger_a\hat{M}_a\neq k\mathbb{I}$, where $k$ is some proportionality constant. Hence, the density operators of \textbf{5}-\textbf{6} (and \textbf{7}-\textbf{8}) are not equivalent to those of Eq.~\eqref{e.EquivClassC}. More generally, by defining the equivalence classes
\begin{equation}
     \varrho_{\mathcal{P}^+_A\setminus\mathcal{P}^+_B}=\{\rho ~|~ \Tr[\rho\hat{L}]=\Tr[\rho_{\bf 6}\hat{L}],~\forall\hat{L}\in\mathfrak{O}(\mathcal{P}^+_A\setminus\mathcal{P}^+_B)\}
    \label{e.EquivClassA/B}
\end{equation}
and
\begin{equation}
    \varrho_{\mathcal{P}^+_B\setminus\mathcal{P}^+_A}=\{\rho ~|~ \Tr[\rho\hat{L}]=\Tr[\rho_{\bf 8}\hat{L}],~\forall\hat{L}\in\mathfrak{O}(\mathcal{P}^+_B\setminus\mathcal{P}^+_A)\}
    \label{e.EquivClassB/A}
\end{equation}
we get that the expectation values of operators having support in $\mathcal{P}^+_A\setminus\mathcal{P}^+_B$ and $\mathcal{P}^+_A\setminus\mathcal{P}^+_B$ evaluated by means of a representative of \eqref{e.EquivClassA/B} and \eqref{e.EquivClassB/A} respectively are different from those obtained by using a representative of $\varrho_{\mathcal{S}_{AB}}$. Therefore, observers that enter $\mathcal{P}^+_A\setminus\mathcal{P}^+_B$ or $\mathcal{P}^+_B\setminus\mathcal{P}^+_A$ cannot use the density operators of \eqref{e.EquivClassC} and get the right statistics.

\subsection{\texorpdfstring{$C\subset\mathcal{P}^+_A\cap\mathcal{P}^+_B$}{C in P[A] cap P[B]}}
\label{s.third_case}
Finally, we consider the case in which the detector $C$ is in the causal future of both $\mathcal{P}^+_A$ and $\mathcal{P}^+_B$, i.e. $C\subset \mathcal{P}^+_A\cap\mathcal{P}^+_B$. In this case, the contextual state of knowledge about the field can be one of the above \textbf{1-8} or
\begin{enumerate}[resume=poss1, label=\textbf{\arabic*)}]
    \item $\mathcal{O}_C$ knows the outcomes of the measurements performed by $\mathcal{O}_A$ and $\mathcal{O}_B$, which are $a$ and $b$ respectively\\$\Rightarrow$ they describe the field's state as $\rho_{\bf 9}=\rho^{A(a)B(b)}_{\mathcal{P}^+_A\cup\mathcal{P}^+_B}$;
\end{enumerate}
However, no density operator in the equivalence classes $\varrho_{\mathcal{S}_{AB}}$, $\varrho_{\mathcal{P}^+_A\setminus\mathcal{P}^+_B}$ or $\varrho_{\mathcal{P}^+_B\setminus\mathcal{P}^+_A}$ is able to give the same expectation values as those obtained by $\rho_{\bf 9}$, i.e.
\begin{equation}
    \langle \hat{L}\rangle_{\bf 9}\neq\langle \hat{L}\rangle_{\bf j}~,\forall {\bf j}=1,\dots,8
    \label{e.9isUnique}
\end{equation}
The proof of this fact is similar to that of the above paragraph, and it can be found in App.~\ref{a.proof5689}. As a consequence, we define a new equivalence class
\begin{equation}
    \varrho_{\mathcal{P}^+_B\cap\mathcal{P}^+_A}=\{\rho ~|~ \Tr[\rho\hat{L}]=\Tr[\rho_{\bf 9}\hat{L}],~\forall\hat{L}\in\mathfrak{O}(\mathcal{P}^+_B\cap\mathcal{P}^+_A)\}
\end{equation}
describing all the states that one can assign to the field while being in $\mathcal{P}^+_B\cap\mathcal{P}^+_A$.

\subsection{Addressing the first issue}
\label{s.density states}
The result of the above discussion is that density operators are not completely contextual to the observers using them. We can imagine observers who do not know the outcomes of the measurements performed in their past light cone, thus getting incorrect results when using non-updated density operators. As a remedy, we assign an equivalence class of density operators that can be interpreted as those contextual to the MIO of Sec.~\ref{s.issue1}. This prescription removes the epistemic uncertainty inherent in contextual states and thus solves the first issue. In this way, we may assign an equivalence class of density operators to each spacetime region belonging to the causal future of measurements, or to regions spacelike to the interaction between detectors and the field. 
 
Let us analyze this prescription in more detail. Assuming an initial state for the field (in our case, the vacuum $\rho_0$), each equivalence class contains a state obtained by the initial one via a map $\mathcal{E}^{(\textbf{a})}$ depending on all measurements' outcomes $\textbf{a}$ performed in the causal past of the region related to the class. Therefore, using these equivalence classes means removing the epistemic uncertainty on the field's state due to the subjective ignorance of non-MIO observers, while keeping the ontic one associated with the causality constraints of relativistic QM.

In other words, we define the equivalence class of states
\begin{equation}
    \varrho_{R}=\{\rho~|~\Tr[\rho\hat{L}]=\Tr[\mathcal{E}^{(\textbf{a})}(\rho_0)\hat{L}]~,~\forall\hat{L}\in\mathfrak{A}(R)\}
    \label{e.equivalenceclass}
\end{equation}
where all measurements giving the outcomes $\textbf{a}$ belong to $\mathcal{J}^-(R)$, and $\mathfrak{A}(R)$ is the set of all observables localised in the region $R$ considered. The map $\mathcal{E}^{(\textbf{a})}$ is directly constructed from the measurements and their outcomes. The simplest case is a single measurement $M$, with measurement operators $\hat{M}_i$ corresponding to the possible outcomes $i$, and Kraus operators 
\begin{equation}
    \hat{K}_i=\bra{\gamma}\hat{U}^\dagger_I\ket{i}~,
    \label{e.Kraus}
\end{equation}
where $\ket{\gamma}$ is the inital state of the detector, and $U_I$ is the time-evolution operator for the interaction between the detector and the field. In the case of selective measurements, the map $\mathcal{E}^{(\textbf{a})}$ with outcome ${\bf a}=(i)$ is defined as  
\begin{equation}
    \mathcal{E}^{(i)}(\rho_\phi)=\frac{\hat{K}_i\rho_\phi\hat{K}^\dagger_i}{p_i} \ , 
\end{equation}
where
\begin{equation}
    p_i=\sqrt{\Tr[\hat{M}_i\hat{U}_I(\rho_{\phi}\otimes\sigma_\gamma)\hat{U}_I^{\dagger}\hat{M}^\dagger_i]}~,
\end{equation}
while, for a non-selective measurement, the map is
\begin{equation}
    \mathcal{E}(\rho_\phi)=\sum_ip_i\mathcal{E}^{(i)}(\rho_\phi)=\sum_i\hat{K}_i\rho_\phi\hat{K}^\dagger_i \ .
\end{equation}
If more measurements are in the causal past of the considered region $R$, the appearing order of the Kraus operators in $\mathcal{E}^{({\bf a})}$ is determined by the causal ordering of the measurement regions in spacetime. Specifically, once we restrict to the case where all measurements have a well-defined causal relation, the causal structure of spacetime induces a partial order $\preceq$ between measurements: for any two given measurements $M_1$ and $M_2$, $M_1\preceq M_2$ iff 
$x\notin\mathcal{J}^+(M_2)$ holds $\forall x\in M_1$. When $M_1\preceq M_2$, we say that $M_2$ is later than $M_1$. Then, for any two measurements $M_1,M_2\in\mathcal{J}^-(R)$, it must either be $M_1\preceq M_2$ or $M_2\preceq M_1$, and the order of appearance of the Kraus operators is chronological, e.g. $\mathcal{E}^{(a_1, a_2)}(\rho)\propto \hat{K}_{a_2}\hat{K}_{a_1}\rho\hat{K}^\dagger_{a_1}\hat{K}^\dagger_{a_2}$ whenever $M_1\preceq M_2$.

Assigning an equivalence class of density operators instead of a single representative to a system is not new; the procedure is similar to assigning to systems vectors up to a global phase factor. The novelty of this proposal is in how the equivalence classes are constructed and the fact that they depend on the spacetime region in which the system is described in a way similar to that proposed by Hellwig and Krauss in Ref.~\cite{HellwigK70}. 

Before proceeding, we stress that we did not address the question of how to assign equivalence classes to the regions containing points that are lightlike to measurements and how the transition from a pre-measurement to a post-measurement equivalence class is realised; these questions are beyond the scope of this article.

\section{Detector-based update of algebraic field states}
\label{s.algebraic states}
In the previous section, we assigned an equivalence class of density operators to each region $R$ having a definite causal relation to all measurements. While solving the first issue, these equivalence classes leave the second unaddressed. This is because the prescription cannot answer questions about non-local quantum correlations and single-shot experiments (as briefly discussed at the end of Sec.~\ref{s.first_case}). To proceed further, we consider a more general version of QFT, called Haag–Kastler axiomatic framework for quantum field theory or, more simply, algebraic quantum field theory (AQFT). 

The choice to turn to AQFT has two additional precursors. First, the connection between detector-based measurements and AQFT was hinted at by PGM in \cite{Polo-GomezEtAl21}. Indeed, they already noted that to avoid the contextuality inherent in the description using density operators it was preferable to move towards a description based on $n$-point functions. Also, they argued for focusing on what experimenters can possibly measure as opposed to prioritizing the Hilbert space of QFT; this is the usual approach of AQFT. However, we find that the $n$-point functions obtained by PGM are still as contextual as the density operators used to calculate them (see Sec.~\ref{s.issue2}). 
Therefore, to render the measurement proposal intrinsically non-contextual, we take a step further and fully adopt the tools of AQFT. Our second motivation for the use of AQFT is that algebraic states are also a crucial part of the FV proposal for measurements in QFT~\cite{FewsterV18, FewsterV23}.


The reason why density operators are not enough to solve the second issue is that, in the previous discussion, we always considered the expectation values of observables belonging to the sets $\mathfrak{A}(R)$ of all observables whose support is contained in a region $R$ having a definite causal relation with all measurements. As it is clear, products of field operators evaluated on $n$ separate points do not always fall in this category, and this led to the inconsistencies presented in Sec.~\ref{s.issue2}. 
As we will see, algebraic states naturally solve this issue, providing a rule for computing generic $n$-point functions. 

\subsection{The GNS construction}
The primitive elements of Haag-Kastler's AQFT are the spacetime $\mathcal{M}$ and a unital $\ast$-algebra $\mathfrak{A}(\mathcal{M})$ defined on it, called the algebra of observables~\cite{HaagEtAl70,Haag92}. To each open causally convex\footnote{An open set $\mathcal{O}\subset\mathcal{M}$ is casually convex iff $\mathcal{O}$ contains all timelike geodesics from $x$ to $y$, for all $x,y\in\mathcal{O}$~\cite{Penrose72}.} bounded region $\mathcal{O}\subset M$ we associate a subalgebra $\mathfrak{A}(\mathcal{O})$ of $\mathfrak{A}(\mathcal{M})$, called the algebra of local observables (see Ref.~\cite{FewsterR19} for a pedagogical review, and Ref.~\cite{Haag92} for further details). In particular, we require that
\begin{equation}
    \mathfrak{A}(O_1)\subset \mathfrak{A}(O_2)
\end{equation}
for all $O_1\subset O_2$ (isotony), and 
\begin{equation}
    [\mathfrak{A}(O_1),\mathfrak{A}(O_2)]=0
\end{equation}
whenever $O_1$ and $O_2$ are spacelike separated (Einstein causality). In general, the elements of $\mathfrak{A}(O)$ have compact support in spacetime~\cite{StreaterW00}.

In this setting, a state $\omega$ is an element of the space $\mathfrak{N}(\mathcal{M})=\mathfrak{A}^*_{+\mathbb{I}}(\mathcal{M})$ of all normalised positive linear forms on $\mathfrak{A}(\mathcal{M})$, and $\omega(A)$ reads as the expectation value of $A\in\mathfrak{A}(\mathcal{M})$ over $\omega$. As it is clear, this notion of a state is very different from the one used in QM and QFT. Yet, it provides all the operationally accessible elements we expect from states of a physical quantum theory (namely, probability distributions related to measurement outcomes and all their momenta). The usual vector states are recovered from the algebraic ones via the so-called GNS construction. This is a way to associate any state $\omega$ with a triple $( \pi_\omega, \mathcal{H}_\omega,\psi_\omega)$, composed by a representation $\pi_\omega$ of $\mathfrak{A}(\mathcal{M})$ on the Hilbert space $\mathcal{H}_\omega$ and a vector $\ket{\psi_\omega}\in\mathcal{H}_\omega$, such that
\begin{equation}
    \omega(A)=\bra{\psi_\omega}\pi_\omega(A)\ket{\psi_\omega}~.
\end{equation}
Starting from this, one can evaluate the expectation values of observables over other vector states by
\begin{equation}
    \omega_{\psi'}(A)=\omega(B^*AB)=\bra{\psi'}\pi_\omega(A)\ket{\psi'}
    \label{e.vector folium}
\end{equation}
where $\ket{\psi'}=\pi_\omega(B)\ket{\psi_\omega}$ is a vector of the representation $\pi_\omega$. In particular, we may assume $\pi_\omega$ cyclic with $\ket{\psi_\omega}$ as the cyclic vector, meaning that any vector state of $\mathcal{H}_\omega$ can be approximated by $\pi_\omega(B)\ket{\psi_\omega}$, for some $B\in\mathfrak{A}(\mathcal{M})$. Moreover, for any given GNS representation of some state $\omega$, one can consider the states $\omega_\rho$ defined by
\begin{equation}
    \omega_\rho(A)=\Tr[\rho\pi_\omega(A)]~,
    \label{e.mixed folium}
\end{equation}
where $\rho$ is a positive trace class operator in the space $\mathcal{B}(\mathcal{H}_\omega)$ of bounded linear operators over $\mathcal{H}_\omega$. The set of all states that can be represented in this way is called the \textit{folium} of $\omega$, and the states contained in it are called \textit{normal states}. Taking a state (either a density operator or a vector state) of the folium of $\omega$ and making the GNS construction on it will change the Hilbert space, the representation of the selected state\footnote{In particular, the state will be a vector state in the new representation, regardless of whether it was a density operator or not in the old one.} and of the algebra's elements, but not the resulting expectation values.

\subsection{A quantum mechanical interlude}
\label{s.Bell_states_collapse}
Before proceeding, let us comment on how the algebraic toolbox describes the measurement-induced collapse of two non-relativistic spacelike systems in an entangled state. To this end, let us consider the Bell state $\ket{\beta_{00}}$ of two qubits $A$ and $B$ placed at some spacelike separated locations. In algebraic terms, this can be read as a system associated with some unital $*$-algebra of observables represented by a tensor product of Hilbert spaces $\mathcal{H}_{AB}=\mathbb{C}^2_A\otimes\mathbb{C}^2_B$, and a state
\begin{equation}
    \omega(AB)=\bra{\beta_{00}}\hat{A}\otimes\hat{B}\ket{\beta_{00}}
    \label{e.algebraic_Bell}
\end{equation}
on $\mathcal{A}$. Thanks to the correlations encoded in Bell states, the above state (maximally) violates the Bell inequality, given in Ref.~\cite{Baez87} for QM in the algebraic framework and in Ref.~\cite{SummersW87a} for AQFT. The state of one qubit, say $B$, changes non-locally when we measure the state of the other; in fact, the L\"uders rule postulates the non-local update
\begin{equation}
    \omega(\bullet)\mapsto\omega^{(m)}(\bullet)=\frac{\omega(M_m\bullet M^*_m)}{\omega(M_mM^*_m)}~,
    \label{e.algebraic_Bell_collapse}
\end{equation}
where $M_m$ are the elements of $\mathcal{A}_A$ associated with the operators defined on $\mathbb{C}_A$ describing a measurement $M$ with outcomes $\{m\}$~\cite{CliftonH01}. The degree of entanglement is generally reduced by the measurement, and the degree of violation of the Bell inequality is lessened accordingly.

In a causal diagram, the state \eqref{e.algebraic_Bell_collapse} is to be used both in $\mathcal{P}^+_M$ and $\mathcal{S}_M$, but not in the past of $M$. This is the picture seen by an observer who post-selects on the specific outcome obtained, and it can be used to describe single-shot experiments. On the contrary, an observer lying outside $\mathcal{J}^+(M)$, having ontic ignorance on the system's state and interpreting the algebraic state as their expectation value of some observable $B$ over many experiments, will use the average state
\begin{equation}
    \omega_{{\rm out}}(B)=\sum_m p_m \omega^{(m)}(B)~.
\end{equation}
Thanks to Einstein locality, the linearity of $\omega$, and that each repetition of $M$ will output the result $m$ with probability $p_m=\omega(M^*_mM_m)$, this state equals the pre-measurement state $\omega(B)$, hence enforcing the result of Sec.~\ref{s.first_case}.

In the following, we will use this example to prescribe an update rule for detector-based measurements in QFT that can address all previously presented issues, i.e. that accounts for quantum non-local correlations and does not demote the state to a contextual quantity.

\subsection{Updated states via AQFT}
\label{s.Equivalence_via_AQFT}
We are now ready to move from equivalence classes of density operators to algebraic states, and propose an update rule for the latter. To do so, we start from the GNS construction of the vacuum state $\omega_0$, i.e. $(\pi_0,\mathcal{H}_0,\ket{0})$, and notice that all operators in the above sections can be expressed as
\begin{equation}
    \hat{O}=\pi_0(O)
\end{equation}
for some local observable $O\in\mathfrak{A}(\mathcal{M})$. This is because, by working on the Hilbert space where the vacuum is represented by $\ket{0}$, we were implicitly assuming the GNS construction related to $\omega_0$. Then, the density operators $\rho$ identifying the classes $\varrho_R$ represent states in the folium of $\omega_0$, and the equivalence classes \eqref{e.equivalenceclass} are included in the equivalence classes
\begin{equation}
    \Omega_{R}=\{\omega\in\mathfrak{N}(R)~|~\omega(L)=\Tr[\mathcal{E}^\textbf{a}(\rho_0)\pi_0(L)]\}
\end{equation}
of (algebraic) states over $\mathfrak{A}(R)$, where $R$ is a causally convex region included in one of the regions identified by the measurements (e.g. $R\subset\mathcal{S}(M_1,\dots,M_n)$).
Before proceeding, it is essential to make two important observations. First, notice that the same local state in a given equivalence class can be extended to many inequivalent global states, meaning that it is only a tool for investigating the local properties of the various regions defined by the measurement. Second, a global state obtained by extending one of the local states obtained by the above procedure is generally not well-behaved on the boundary of the regions identified by the measurements; in particular, it was proven in Ref.~\cite{FewsterV13} (corollary 3.3 therein) that a normal state in a double cone\footnote{A double cone is defined as the set of all points lying on smooth
timelike curves anchored at two timelike events~\cite{FewsterR19}} $\mathcal{D}$ becomes non-normal when restricted to a double cone $\mathcal{D}'\subset\mathcal{D}$, and normal states locally defined on $\mathcal{D}'$ extend to states on $\mathcal{D}$ that diverge on the boundary of $\mathcal{D}'$. While the regions in our construction are not double cones, we still expect a diverging behaviour to appear on the measurement light cone which may be justified by a discontinuous transition induced by the L\"uders rule. Because of these issues, the above local states cannot be extended to global states. Yet, as this state reproduces the results of \cite{Polo-GomezEtAl21} without the need for a processing region, it can be seen as the contextual state of the MIO and a faithful tool to describe expectation values of observables localized in one region $R$. In particular, computing $n$-point functions within a given region $R$ by these states does not require identifying the related processing region, hence removing the epistemic ignorance from the state assignment.

As the state constructed above cannot be used to test correlations across measurement regions (or account for the results of single-shot experiments), we now ask whether we can find a state over $\mathcal{M}$ that correctly provides the expectation value of the product of two local observables $L_1$ and $L_2$ supported on regions $R$ and $R'$ not belonging to the same measurement-defined region. To this end, we consider a measurement $M$ with outcome $m$ and rewrite Eq.~\eqref{e.contextAnoB} as
\begin{equation}
    \rho^{M(m)}_{\mathcal{P}^+_M\cup\mathcal{S}_M}=\frac{\hat{K}_m\rho_{\phi_0}\hat{K}_m^\dagger}{\Tr[\hat{K}_m\rho_{\phi_0}\hat{K}_m^\dagger]}
    \label{e.pre_algebraic_assignment}
\end{equation}
where $\hat{K}_m=\bra{m}\hat{U}_A\ket{\alpha}$, and where we extended the validity of \eqref{e.contextAnoB} to $\mathcal{S}_M$ because of the entanglement properties of the vacuum \cite{SummersW87a, SummersW87b} and our above discussion about updates on entangled quantum states. The expectation values computed by the density operator \eqref{e.pre_algebraic_assignment} can be equivalently evaluated by the algebraic state
\begin{equation}
    \omega^{(m)}(\bullet)=\frac{\omega_0(K_m\bullet K_m^*)}{\omega_0(K_mK_m^*)}~,
    \label{e.QFT_update}
\end{equation}
where the elements $K_m$ of $\mathfrak{A}(M)$ are represented as $\hat{K}_m$ through the GNS construction of the vacuum. The above state is related to a local selective preparation $T:\mathfrak{A}(\mathcal{M})\rightarrow\mathfrak{A}(\mathcal{M})$ defined by $O\mapsto K_m O K_m^*$ \cite{BuchholzEtAl86, Werner87}. In practice, using a detector induces the definition of a set of Kraus operators on the local algebra of the field; these operators are localized in the measurement region $M$, and hence belong to $\mathfrak{A}(M)$. The update is obtained by considering the selective operation generated by applying a single Kraus operator from this set~\cite{CliftonH01}.

In practice, our proposal consists of using the detector-based measurement scheme of PGM to build, through its Kraus representation, a non-selective measurement protocol acting on the local algebra $\mathcal{A}(M)$ of the field. On the one hand, the fact that field states are inherently entangled means that local operations modify the state of the field outside the causal future of the measurement region, and hence affect the outcomes of single-shot experiments in all spacetime. On the other hand, averaging over many \textit{i.i.d.} experiments gives a localized non-selective operation, for which the results obtained by PGM in the exterior of $\mathcal{J}^+(M)$ still hold.

\subsection{The importance of knowing the outcomes}
As an example, let us use the algebraic framework to show that any naive extension of the states of $\Omega_{\mathcal{S}_M}$ to $\mathcal{P}^+_M$ gives the wrong global state. To demonstrate this, we consider the case 5 of Sec.~\ref{s.second_case}, where a detector $M=A$ measures the outcome $a$, known in the causal future region $\mathcal{P}^+_M$, with the appropriate equivalence class corresponding to the state $\rho_{\bf 5}$ given by \eqref{e.contextAnoB}. In turn, in the regions of $\mathcal{S}_M$, the appropriate state is a restriction of the vacuum $\rho_0=\ket{0}\bra{0}$.
We start by noticing that the corresponding algebraic states $\omega_0$ and $\omega_{\rho_{\bf 5}}$ are both normal states of the representation $\pi_0$, i.e. 
\begin{equation}
    \omega_0(L)=\bra{0}\hat{L}\ket{0}
    \label{e.omega_rho_0}
\end{equation}
with $\hat{L}=\pi_0(L)$, and $L\in\mathfrak{A}(\mathcal{S}_M)$, and
\begin{equation}
    \omega_{\rho_{\bf 5}}(L)=\bra{0}\hat{K}_a^\dagger\hat{L}\hat{K}_a\ket{0}
    \label{e.omega_rho_5}
\end{equation}
with $L\in\mathfrak{A}(\mathcal{P}^+_M)$; the operator \eqref{e.Kraus} can be obtained as the $\pi_0$-representation of a Kraus operator $K_a$ of the algebra $\mathfrak{A}(M)\subset \mathfrak{A}(\mathcal{M})$~\cite{Kraus83}. Extending Eq.~\eqref{e.omega_rho_0} to hold for elements in $\mathfrak{A}(\mathcal{P}^+_M)$, we can use Einstein's causality and the fact that $\pi_0$ is a group homomorphism to get that
\begin{equation}
    \omega_0(L)=\omega_{\rho_{\bf 5}}(L)~,~\forall L\in\mathfrak{A}(\mathcal{S}_M)
\end{equation}
holds if and only if
\begin{equation}
    \pi_0(K_a^*K_a)=\mathbb{I}_{\mathcal{H}_0}~,
\end{equation}
which means
\begin{equation}
    K_a^*K_a=\mathbb{I}_{\mathfrak{A}(M)}~.
    \label{e.Q_condition}
\end{equation}
However, Eq.~\eqref{e.Q_condition} does not hold for the Kraus operators defined in Eq.~\eqref{e.Kraus} and, more generally, for any $K_a\in\mathfrak{A}(\mathcal{M})$ associated with an informative measurement~\cite{Ozawa12, OkamuraO15}. Hence the incorrect attempt to use $\omega_0(L)$ in $\mathcal{P}^+_M$ gives wrong expectation values.

This result is the algebraic version of the one obtained in Sec.~VIII of Ref.~\cite{Polo-GomezEtAl21} and in Sec.~\ref{s.state_assignment} of this manuscript, where we proved that different (contextual) density operators do not always give the same expectation values.

\subsection{Addressing the second issue}
\label{s.second_solution}
We now show how our algebraic state assignment consistently describes one- and two-point functions, solving the issue presented in Sec.~\ref{s.issue2}. In particular, for a given measurement $M$ we focus our discussion about two-point functions on field operators evaluated at any one and two points in $\mathcal{P}^+_M$ and $\mathcal{S}_M=\mathcal{M}\setminus\mathcal{J}(M)$. The construction can be generalised to more measurements and higher order $n$-point functions evaluated in any collection of points belonging to regions having definite causal relation with all measurements.

For the rest of this section, we will consider point-wise $n$-point functions, i.e. Wightman functions evaluated on spacetime points. However, these must be read via the nuclear theorem as a tool for calculating the smeared Wightman functions
\begin{equation}
    W(f_1,\dots,f_n)=\int \prod_{j=1}^ndx_j f_j(x_j)\langle \phi(x_1)\dots\phi(x_n)\rangle~,
    \label{e.smeared_W}
\end{equation}
where $f_i$ are Schwartz test functions localised in compact regions with definite causal relation with all measurements~\cite{StreaterW00}.

\subsubsection{one-point functions}
First, suppose $M$ gives the outcome $m$ and study the related one-point functions. By our prescription, these are
\begin{equation}
    \braket{\hat{\phi}(x)}=\omega^{(m)}(\phi(x))=\frac{\omega_0(M_m\phi(x)M^*_m)}{\omega_0(M_m M^*_m)}
\end{equation}
if $x\in\mathcal{P}^+_M$ or $\mathcal{S}_M$, and 
\begin{equation}
    \braket{\hat{\phi}(x)}=\omega_0(\phi(x))
\end{equation}
if $x\in\mathcal{P}^-_M$. Yet, if we focus on the expectation value of observables localized in $R\subset\mathcal{S}_M$ over many (non-selected) \textit{i.i.d.} measurements in $M$, then we get
\begin{equation}
    \braket{\hat{\phi}(x)}=\sum_m p_m\omega^{(m)}(\phi(x))=\omega_0(\phi(x))
\end{equation}
as in Sec.~\ref{s.Bell_states_collapse}. Hence, in the non-selective scenario, the rule proposed in Ref.~\cite{Polo-GomezEtAl21} for one-point functions can be applied in our framework without modification.

\subsubsection{two-point functions}
Next, we consider two-point functions and, following the same logic as above, we get
\begin{equation}
    \braket{\hat{\phi}(x)\hat{\phi}(y)}=
    \omega^{(m)}(\phi(x)\phi(y))
\end{equation}
for any choice of $x$ and $y$ in $R\subset\mathcal{P}^+_M$ or $R\subset\mathcal{S}_M$. Again, considering the average over many samples of two-point functions supported in a region $R$ outside $\mathcal{P}^+_M$ gives a non-selective rule equivalent to that of PGM. Notice that, while our assignment is identical to that of Ref.~\cite{Polo-GomezEtAl21} in the region $\mathcal{P}^+_M$, it essentially differs from it for points in $R\subset\mathcal{S}_M$, as there PGM always assign a non-selective update. Moreover, here states are non-contextual quantities, and the resulting $2$-point functions do not depend on the region where they are calculated in (i.e. the location of the processing region), hence addressing the issue raised by the discussion of Sec.~\ref{s.issue2} and exemplified by Fig.s~\ref{f.issue2b} and \ref{f.issue2c}.

Let us now discuss some consequences of this prescription. As a first case, we consider $x,y\in R\subset\mathcal{P}^+_M$. This reads
\begin{equation}
    \braket{\hat{\phi}(x)\hat{\phi}(y)}=\omega_0(K_m^*\phi(x)\phi(y)K_m)\neq \omega_0(\phi(x)\phi(y))~,
    \label{e.two-points_w_kraus}
\end{equation}
in accordance with the prescription of PGM. The same holds if $x\in\mathcal{P}^+_M$ and $y\in\mathcal{S}_M$ (and viceversa). This result is also identical to PGM's, but for different reasons. In our prescription, this comes from the ontic nature of the state across all spacetime, hence making the correlation a property of the field. Finally, Eq.~\eqref{e.two-points_w_kraus} also holds if $x,y\in R\subset\mathcal{S}_M$; this contrasts with PGM's prescription, but solves all the pathological cases presented in Sec.~\ref{s.issue2}. Moreover, the result of PGM is recovered when thinking of the two-point function measured over many \textit{i.i.d.} experiments, where no post-selection procedure is applied. In this case, 
\begin{equation}
\begin{split}
    \braket{\hat{\phi}(x)\hat{\phi}(y)}_{NS}=\omega_0(\phi(x)\phi(y))~,
    \label{e.non_selective_two_points}
\end{split}
\end{equation}
as it is easy to show by linearity and Einstein causality. Notice that, even if the above expectation value resembles that of the vacuum outside $\mathcal{J}(M)$, the global state is non-trivial to construct and interpret by standard textbook QFT tools.

Before concluding, let us notice that Eq.~\eqref{e.non_selective_two_points} is obtained by restricting the selective states to some regions $R$ outside $\mathcal{J}(M)$ and average overall different outcomes. This procedure is possible because, assuming the spacetime dimensions to be $d>2$, it is always possible to find an open (and causally convex) set $R\subset \mathcal{S}_M$ covering any two points outside $\mathcal{J}(M)$ over which a state $\omega$ can be defined (the proof of this fact can be found in App.~\ref{a.open_causally_convex_set}). Yet, it is interesting to understand why this is not true in $d=2$ and the resulting implications. In this case, $\mathcal{J}(M)$ cuts the spacetime into two disconnected pieces, hence not always allowing a definition of an open set enclosing all points over which the two-point function is evaluated. For example, this is impossible when the points are selected as in Fig.~\ref{f.issue2b}. Discussing the $d=2$ case is not trivial, and we leave in-depth dissemination for future work; here, we just remark that Eq.~\eqref{e.non_selective_two_points} does not yield the usual vacuum expectation value of the two-point function in this scenario and that density operators prove to be inadequate for addressing this situation, highlighting an advantage of the algebraic framework over the standard approach.

\subsubsection{\texorpdfstring{Measuring $n$-point functions}{Measuring n-point functions}}

Before concluding, let us briefly comment on how we can use an operational procedure to measure the above correlation functions. First, we distinguish between the standard notion of a correlation function as a mathematical object as defined and computed in introductory QFT textbooks and a ``measured'' correlation function. By the latter, we mean measuring $\bra{\psi} \phi(y) \phi(x) \ket{\psi}$ directly via some properly designed field-detector coupling and a POVM defined on the latter, which serves to directly measure the field values $\phi (y),\phi (x)$ when the field is prepared in the state $\ket{\psi}$ and the expectation value from the measurement data. The details of which interaction and POVM are the best for performing such measurements are beyond the scope of this paper; we refer to Ref.~\cite{FewsterV18} for an example of such a realization in the case of field-like probes. When discussing measured quantities (specifically, expectation values), we need a way to perform many independent and identically distributed (\textit{i.i.d.}) experiments upon the system of interest. This can be achieved through many copies of the measured field or by measuring the field by many detectors at many sets of spacelike separated points. While the first option can be considered unfeasible on practical grounds, the second is available only if the field and the initial state we want to measure are Poincar\'e invariant~\cite{BostelmannEtAl21, HellwigK70}. Assuming the latter condition, one should design proper measurement configurations such that the outcomes obtained from the collection of detectors give reliable and valuable information about the selected field observable (in our case, $n$-point functions). Yet, it is crucial to notice that, even by this procedure, the measurements performed on the probes change the state of the field, meaning that the obtained expectation value gives valuable information about the past state, which cannot be used for future predictions.

In particular, this causes problems when trying to measure $n$-point functions at timelike separated points. Suppose an observer wants to test the field in two points $x$ and $y$, with $y\in\mathcal{J}^+(x)$. Then, when getting information about $\phi(x)$ via a measurement on the probe, the observer collapses the state of the field in the future of $x$, destroying information and making the original unmeasured field state at $y$ inaccessible. Therefore, our discussion cannot apply to any $n$-point function for which any two points are timelike to each other. Nonetheless, it should be noted that measuring timelike correlations is still possible; this can be done by employing either several spacelike separated detectors or many replicated spacetimes (or local experimental setups), hence accessing a collection of \textit{i.i.d.} experiments from which one can measure the associate probabilities and thus reconstruct the required $n$-point function describing the field state in the absence of measurements (a discussion on why replicas are needed when repeatedly measuring the same system can be found in Ref.~\cite{PranziniEtAl23}). Indeed, once the required data is gathered by the above procedure, one can use, for example, the prescription of Ref.~\cite{deRamonEtAl18} and obtain any correlation function, including those at timelike separated points. 

\section{Conclusions}
\label{s.conclusions}
In this work, we argued against the contextuality (read, subjectivity) of QFT states and proposed interpreting states as non-contextual spacetime-dependent objects. The two interpretations can be reconciled by considering the state described by a maximally informed observer. However, the first cannot account for non-local measurement-induced state updates, for which a non-contextual state is required. Starting from PGM's observation that density operators are contextual states with a hidden spacetime dependence, we argued for a shift in the description of the detector-based measurement scheme in which density operators are replaced by more general algebraic states. In particular, we defined our spacetime-dependent states by assigning an equivalence class of density operators to each spacetime region having a definite causal relation with all measurements performed in spacetime, and then distilled an algebraic state from the collection of these classes. The need for using algebraic states came from the fact that, while clarifying the spacetime dependence, equivalence classes are insufficient to consistently evaluate the expectation values of non-local observables (e.g. $n$-point functions). Our construction enables the definition of Kraus operators acting on the local algebra of the field, and hence to evaluate expectation values of spacelike non-local observables across measurements' lightcones. Finally, we employed this construction to calculate two-point functions over those choices of points across measurement lightcones that were pathological when using the original notion of contextual states.

On the one hand, PGM's framework~\cite{Polo-GomezEtAl21} shows potential for widespread adoption due to its formal simplicity, close connection to conventional techniques of non-relativistic quantum mechanics and QFT, and a clear identification of measurement apparatuses with non-relativistic probes such as Unruh-DeWitt detectors. On the other hand, we showed the algebraic toolbox is well suited to describe measurements and the related non-local state update, allowing a direct way to compute the expectation value of non-local observables. However, using algebraic tools (e.g. in the form of the FV proposal~\cite{FewsterV23}) presents some challenges because 1) it requires a more sophisticated level of familiarity with AQFT and 2) it uses an auxiliary quantum field to probe the field of interest. For the auxiliary field, an implementation strategy in terms of measurable systems is less transparent due to the challenges of designing an experimental realization of the required measurement scheme. Through our work, we have attempted to construct a sufficiently simple framework that amalgamates key features from the proposals of PGM and FV, hence solving the issues of the former by the tools of the latter.

\section*{Acknowledgments}
The authors gratefully thank Jos\'e Polo-G\'omez and Eduardo Mart\'in-Mart\'inez for discussions and constructive criticism. N.P. also thanks Claudio Dappiaggi, Paolo Meda, and Paolo Perinotti for useful remarks and discussions. Their help greatly improved the quality of our work, and any inaccuracies contained herein are solely the authors' responsibility. The authors thank the Wallenberg Initiative on Networks and Quantum Information (WINQ) for hospitality at the workshop ``Informational Foundations of QFT'', where a part of this work was completed, and acknowledge the financial support of the Research Council of Finland through the Finnish Quantum Flagship project (358878, UH). N.P. also expresses gratitude for financial support from the Magnus Ehrnrooth Foundation and the Academy of Finland via the Centre of Excellence program (Project No. 336810 and Project No. 336814).

\bibliography{bib}

\begin{thebibliography}{42}%
\makeatletter
\providecommand \@ifxundefined [1]{%
 \@ifx{#1\undefined}
}%
\providecommand \@ifnum [1]{%
 \ifnum #1\expandafter \@firstoftwo
 \else \expandafter \@secondoftwo
 \fi
}%
\providecommand \@ifx [1]{%
 \ifx #1\expandafter \@firstoftwo
 \else \expandafter \@secondoftwo
 \fi
}%
\providecommand \natexlab [1]{#1}%
\providecommand \enquote  [1]{``#1''}%
\providecommand \bibnamefont  [1]{#1}%
\providecommand \bibfnamefont [1]{#1}%
\providecommand \citenamefont [1]{#1}%
\providecommand \href@noop [0]{\@secondoftwo}%
\providecommand \href [0]{\begingroup \@sanitize@url \@href}%
\providecommand \@href[1]{\@@startlink{#1}\@@href}%
\providecommand \@@href[1]{\endgroup#1\@@endlink}%
\providecommand \@sanitize@url [0]{\catcode `\\12\catcode `\$12\catcode `\&12\catcode `\#12\catcode `\^12\catcode `\_12\catcode `\%12\relax}%
\providecommand \@@startlink[1]{}%
\providecommand \@@endlink[0]{}%
\providecommand \url  [0]{\begingroup\@sanitize@url \@url }%
\providecommand \@url [1]{\endgroup\@href {#1}{\urlprefix }}%
\providecommand \urlprefix  [0]{URL }%
\providecommand \Eprint [0]{\href }%
\providecommand \doibase [0]{https://doi.org/}%
\providecommand \selectlanguage [0]{\@gobble}%
\providecommand \bibinfo  [0]{\@secondoftwo}%
\providecommand \bibfield  [0]{\@secondoftwo}%
\providecommand \translation [1]{[#1]}%
\providecommand \BibitemOpen [0]{}%
\providecommand \bibitemStop [0]{}%
\providecommand \bibitemNoStop [0]{.\EOS\space}%
\providecommand \EOS [0]{\spacefactor3000\relax}%
\providecommand \BibitemShut  [1]{\csname bibitem#1\endcsname}%
\let\auto@bib@innerbib\@empty
\bibitem [{\citenamefont {Polo-G\'omez}\ \emph {et~al.}(2022)\citenamefont {Polo-G\'omez}, \citenamefont {Garay},\ and\ \citenamefont {Mart\'{\i}n-Mart\'{\i}nez}}]{Polo-GomezEtAl21}%
  \BibitemOpen
  \bibfield  {author} {\bibinfo {author} {\bibfnamefont {J.}~\bibnamefont {Polo-G\'omez}}, \bibinfo {author} {\bibfnamefont {L.~J.}\ \bibnamefont {Garay}},\ and\ \bibinfo {author} {\bibfnamefont {E.}~\bibnamefont {Mart\'{\i}n-Mart\'{\i}nez}},\ }\bibfield  {title} {\bibinfo {title} {{A detector-based measurement theory for quantum field theory}},\ }\href {https://doi.org/10.1103/PhysRevD.105.065003} {\bibfield  {journal} {\bibinfo  {journal} {Phys. Rev. D}\ }\textbf {\bibinfo {volume} {105}},\ \bibinfo {pages} {065003} (\bibinfo {year} {2022})}\BibitemShut {NoStop}%
\bibitem [{\citenamefont {Sorkin}(1993)}]{Sorkin93}%
  \BibitemOpen
  \bibfield  {author} {\bibinfo {author} {\bibfnamefont {R.~D.}\ \bibnamefont {Sorkin}},\ }\bibfield  {title} {\bibinfo {title} {{Impossible measurements on quantum fields}},\ }in\ \href@noop {} {\emph {\bibinfo {booktitle} {Directions in general relativity: Proceedings of the 1993 International Symposium, Maryland}}},\ Vol.~\bibinfo {volume} {2}\ (\bibinfo {year} {1993})\ pp.\ \bibinfo {pages} {293--305},\ \Eprint {https://arxiv.org/abs/gr-qc/9302018} {arXiv:gr-qc/9302018} \BibitemShut {NoStop}%
\bibitem [{\citenamefont {Peres}\ and\ \citenamefont {Terno}(2004)}]{PeresT04}%
  \BibitemOpen
  \bibfield  {author} {\bibinfo {author} {\bibfnamefont {A.}~\bibnamefont {Peres}}\ and\ \bibinfo {author} {\bibfnamefont {D.~R.}\ \bibnamefont {Terno}},\ }\bibfield  {title} {\bibinfo {title} {{Quantum information and relativity theory}},\ }\href {https://doi.org/10.1103/RevModPhys.76.93} {\bibfield  {journal} {\bibinfo  {journal} {Rev. Mod. Phys.}\ }\textbf {\bibinfo {volume} {76}},\ \bibinfo {pages} {93} (\bibinfo {year} {2004})}\BibitemShut {NoStop}%
\bibitem [{\citenamefont {Florig}\ and\ \citenamefont {Summers}(1997)}]{FlorigS97}%
  \BibitemOpen
  \bibfield  {author} {\bibinfo {author} {\bibfnamefont {M.}~\bibnamefont {Florig}}\ and\ \bibinfo {author} {\bibfnamefont {S.~J.}\ \bibnamefont {Summers}},\ }\bibfield  {title} {\bibinfo {title} {{On the statistical independence of algebras of observables}},\ }\href {https://doi.org/10.1063/1.531812} {\bibfield  {journal} {\bibinfo  {journal} {J. Math. Phys.}\ }\textbf {\bibinfo {volume} {38}},\ \bibinfo {pages} {1318} (\bibinfo {year} {1997})}\BibitemShut {NoStop}%
\bibitem [{\citenamefont {Hellwig}\ and\ \citenamefont {Kraus}(1970)}]{HellwigK70}%
  \BibitemOpen
  \bibfield  {author} {\bibinfo {author} {\bibfnamefont {K.~E.}\ \bibnamefont {Hellwig}}\ and\ \bibinfo {author} {\bibfnamefont {K.}~\bibnamefont {Kraus}},\ }\bibfield  {title} {\bibinfo {title} {{Formal Description of Measurements in Local Quantum Field Theory}},\ }\href {https://doi.org/10.1103/PhysRevD.1.566} {\bibfield  {journal} {\bibinfo  {journal} {Phys. Rev. D}\ }\textbf {\bibinfo {volume} {1}},\ \bibinfo {pages} {566} (\bibinfo {year} {1970})}\BibitemShut {NoStop}%
\bibitem [{\citenamefont {Fewster}\ and\ \citenamefont {Verch}(2020)}]{FewsterV18}%
  \BibitemOpen
  \bibfield  {author} {\bibinfo {author} {\bibfnamefont {C.~J.}\ \bibnamefont {Fewster}}\ and\ \bibinfo {author} {\bibfnamefont {R.}~\bibnamefont {Verch}},\ }\bibfield  {title} {\bibinfo {title} {{Quantum fields and local measurements}},\ }\href {https://doi.org/10.1007/s00220-020-03800-6} {\bibfield  {journal} {\bibinfo  {journal} {Commun. Math. Phys.}\ }\textbf {\bibinfo {volume} {378}},\ \bibinfo {pages} {851} (\bibinfo {year} {2020})}\BibitemShut {NoStop}%
\bibitem [{\citenamefont {Fuchs}(2002)}]{Fuchs02}%
  \BibitemOpen
  \bibfield  {author} {\bibinfo {author} {\bibfnamefont {C.~A.}\ \bibnamefont {Fuchs}},\ }\href@noop {} {\bibinfo {title} {{Quantum Mechanics as Quantum Information (and only a little more)}}} (\bibinfo {year} {2002}),\ \Eprint {https://arxiv.org/abs/quant-ph/0205039} {arXiv:quant-ph/0205039} \BibitemShut {NoStop}%
\bibitem [{\citenamefont {Ruep}(2021)}]{Ruep21}%
  \BibitemOpen
  \bibfield  {author} {\bibinfo {author} {\bibfnamefont {M.~H.}\ \bibnamefont {Ruep}},\ }\bibfield  {title} {\bibinfo {title} {{Weakly coupled local particle detectors cannot harvest entanglement}},\ }\href {https://doi.org/10.1088/1361-6382/ac1b08} {\bibfield  {journal} {\bibinfo  {journal} {Class. Quantum Grav.}\ }\textbf {\bibinfo {volume} {38}},\ \bibinfo {pages} {195029} (\bibinfo {year} {2021})}\BibitemShut {NoStop}%
\bibitem [{\citenamefont {Hawking}\ and\ \citenamefont {Ellis}(2023)}]{HawkingEllis73}%
  \BibitemOpen
  \bibfield  {author} {\bibinfo {author} {\bibfnamefont {S.~W.}\ \bibnamefont {Hawking}}\ and\ \bibinfo {author} {\bibfnamefont {G.~F.~R.}\ \bibnamefont {Ellis}},\ }\href {https://doi.org/10.1017/9781009253161} {\emph {\bibinfo {title} {{The Large Scale Structure of Space-Time}}}},\ Cambridge Monographs on Mathematical Physics\ (\bibinfo  {publisher} {Cambridge University Press},\ \bibinfo {year} {2023})\BibitemShut {NoStop}%
\bibitem [{\citenamefont {Unruh}(1976)}]{Unruh76}%
  \BibitemOpen
  \bibfield  {author} {\bibinfo {author} {\bibfnamefont {W.~G.}\ \bibnamefont {Unruh}},\ }\bibfield  {title} {\bibinfo {title} {{Notes on black-hole evaporation}},\ }\href {https://doi.org/10.1103/PhysRevD.14.870} {\bibfield  {journal} {\bibinfo  {journal} {Phys. Rev. D}\ }\textbf {\bibinfo {volume} {14}},\ \bibinfo {pages} {870} (\bibinfo {year} {1976})}\BibitemShut {NoStop}%
\bibitem [{\citenamefont {DeWitt}(1980)}]{DeWitt80}%
  \BibitemOpen
  \bibfield  {author} {\bibinfo {author} {\bibfnamefont {B.~S.}\ \bibnamefont {DeWitt}},\ }\bibinfo {title} {{Quantum Gravity: the new synthesis}},\ in\ \href {http://inis.iaea.org/search/search.aspx?orig_q=RN:11506258} {\emph {\bibinfo {booktitle} {{General Relativity}: {An Einstein Centenary Survey}}}}\ (\bibinfo {year} {1980})\ pp.\ \bibinfo {pages} {680--745}\BibitemShut {NoStop}%
\bibitem [{\citenamefont {Birrell}\ and\ \citenamefont {Davies}(1984)}]{BirrellD82}%
  \BibitemOpen
  \bibfield  {author} {\bibinfo {author} {\bibfnamefont {N.~D.}\ \bibnamefont {Birrell}}\ and\ \bibinfo {author} {\bibfnamefont {P.~C.~W.}\ \bibnamefont {Davies}},\ }\href {https://doi.org/10.1017/CBO9780511622632} {\emph {\bibinfo {title} {{Quantum Fields in Curved Space}}}},\ Cambridge Monographs on Mathematical Physics\ (\bibinfo  {publisher} {Cambridge Univ. Press},\ \bibinfo {address} {Cambridge, UK},\ \bibinfo {year} {1984})\BibitemShut {NoStop}%
\bibitem [{\citenamefont {Wald}(1995)}]{Wald95}%
  \BibitemOpen
  \bibfield  {author} {\bibinfo {author} {\bibfnamefont {R.~M.}\ \bibnamefont {Wald}},\ }\href@noop {} {\emph {\bibinfo {title} {{Quantum Field Theory in Curved Space-Time and Black Hole Thermodynamics}}}},\ Chicago Lectures in Physics\ (\bibinfo  {publisher} {University of Chicago Press},\ \bibinfo {address} {Chicago, IL},\ \bibinfo {year} {1995})\BibitemShut {NoStop}%
\bibitem [{\citenamefont {Crispino}\ \emph {et~al.}(2008)\citenamefont {Crispino}, \citenamefont {Higuchi},\ and\ \citenamefont {Matsas}}]{CrispinoEtAl08}%
  \BibitemOpen
  \bibfield  {author} {\bibinfo {author} {\bibfnamefont {L.~C.~B.}\ \bibnamefont {Crispino}}, \bibinfo {author} {\bibfnamefont {A.}~\bibnamefont {Higuchi}},\ and\ \bibinfo {author} {\bibfnamefont {G.~E.~A.}\ \bibnamefont {Matsas}},\ }\bibfield  {title} {\bibinfo {title} {{The Unruh effect and its applications}},\ }\href {https://doi.org/10.1103/RevModPhys.80.787} {\bibfield  {journal} {\bibinfo  {journal} {Rev. Mod. Phys.}\ }\textbf {\bibinfo {volume} {80}},\ \bibinfo {pages} {787} (\bibinfo {year} {2008})}\BibitemShut {NoStop}%
\bibitem [{\citenamefont {Schlicht}(2004)}]{Schlicht04}%
  \BibitemOpen
  \bibfield  {author} {\bibinfo {author} {\bibfnamefont {S.}~\bibnamefont {Schlicht}},\ }\bibfield  {title} {\bibinfo {title} {{Considerations on the Unruh effect: causality and regularization}},\ }\href {https://doi.org/10.1088/0264-9381/21/19/011} {\bibfield  {journal} {\bibinfo  {journal} {Class. Quantum Grav.}\ }\textbf {\bibinfo {volume} {21}},\ \bibinfo {pages} {4647–4660} (\bibinfo {year} {2004})}\BibitemShut {NoStop}%
\bibitem [{\citenamefont {Louko}\ and\ \citenamefont {Satz}(2006)}]{LoukoS06}%
  \BibitemOpen
  \bibfield  {author} {\bibinfo {author} {\bibfnamefont {J.}~\bibnamefont {Louko}}\ and\ \bibinfo {author} {\bibfnamefont {A.}~\bibnamefont {Satz}},\ }\bibfield  {title} {\bibinfo {title} {{How often does the Unruh-DeWitt detector click? Regularization by a spatial profile}},\ }\href {https://doi.org/10.1088/0264-9381/23/22/015} {\bibfield  {journal} {\bibinfo  {journal} {Class. Quantum Grav.}\ }\textbf {\bibinfo {volume} {23}},\ \bibinfo {pages} {6321} (\bibinfo {year} {2006})}\BibitemShut {NoStop}%
\bibitem [{\citenamefont {Satz}(2007)}]{Satz07}%
  \BibitemOpen
  \bibfield  {author} {\bibinfo {author} {\bibfnamefont {A.}~\bibnamefont {Satz}},\ }\bibfield  {title} {\bibinfo {title} {{Then again, how often does the Unruh–DeWitt detector click if we switch it carefully?}},\ }\href {https://doi.org/10.1088/0264-9381/24/7/003} {\bibfield  {journal} {\bibinfo  {journal} {Class. Quantum Grav.}\ }\textbf {\bibinfo {volume} {24}} (\bibinfo {year} {2007})}\BibitemShut {NoStop}%
\bibitem [{\citenamefont {Heinosaari}\ and\ \citenamefont {Ziman}(2011)}]{HeinosaariZ11}%
  \BibitemOpen
  \bibfield  {author} {\bibinfo {author} {\bibfnamefont {T.}~\bibnamefont {Heinosaari}}\ and\ \bibinfo {author} {\bibfnamefont {M.}~\bibnamefont {Ziman}},\ }\href@noop {} {\emph {\bibinfo {title} {The Mathematical Language of Quantum Theory: From Uncertainty to Entanglement}}}\ (\bibinfo  {publisher} {Cambridge University Press},\ \bibinfo {year} {2011})\BibitemShut {NoStop}%
\bibitem [{\citenamefont {von Neumann}(1955)}]{vonNeumann55}%
  \BibitemOpen
  \bibfield  {author} {\bibinfo {author} {\bibfnamefont {J.}~\bibnamefont {von Neumann}},\ }\href@noop {} {\emph {\bibinfo {title} {{Mathematical Foundations of Quantum Mechanics}}}},\ Goldstine Printed Materials\ (\bibinfo  {publisher} {Princeton University Press},\ \bibinfo {year} {1955})\BibitemShut {NoStop}%
\bibitem [{\citenamefont {Ozawa}(1984)}]{Ozawa84}%
  \BibitemOpen
  \bibfield  {author} {\bibinfo {author} {\bibfnamefont {M.}~\bibnamefont {Ozawa}},\ }\bibfield  {title} {\bibinfo {title} {{Quantum measuring processes of continuous observables}},\ }\href {https://doi.org/10.1063/1.526000} {\bibfield  {journal} {\bibinfo  {journal} {J. Math. Phys.}\ }\textbf {\bibinfo {volume} {25}},\ \bibinfo {pages} {79} (\bibinfo {year} {1984})}\BibitemShut {NoStop}%
\bibitem [{\citenamefont {Maudlin}(1995)}]{Maudlin95}%
  \BibitemOpen
  \bibfield  {author} {\bibinfo {author} {\bibfnamefont {T.}~\bibnamefont {Maudlin}},\ }\bibfield  {title} {\bibinfo {title} {Three measurement problems},\ }\href {https://doi.org/10.1007/BF00763473} {\bibfield  {journal} {\bibinfo  {journal} {Topoi}\ }\textbf {\bibinfo {volume} {14}},\ \bibinfo {pages} {7} (\bibinfo {year} {1995})}\BibitemShut {NoStop}%
\bibitem [{\citenamefont {Haag}(1992)}]{Haag92}%
  \BibitemOpen
  \bibfield  {author} {\bibinfo {author} {\bibfnamefont {R.}~\bibnamefont {Haag}},\ }\href@noop {} {\emph {\bibinfo {title} {{Local Quantum Physics: Fields, Particles, Algebras}}}},\ Theoretical and Mathematical Physics\ (\bibinfo  {publisher} {Springer-Verlag},\ \bibinfo {year} {1992})\BibitemShut {NoStop}%
\bibitem [{\citenamefont {Fewster}\ and\ \citenamefont {Verch}(2023)}]{FewsterV23}%
  \BibitemOpen
  \bibfield  {author} {\bibinfo {author} {\bibfnamefont {C.~J.}\ \bibnamefont {Fewster}}\ and\ \bibinfo {author} {\bibfnamefont {R.}~\bibnamefont {Verch}},\ }\href@noop {} {\bibinfo {title} {{Measurement in Quantum Field Theory}}} (\bibinfo {year} {2023}),\ \Eprint {https://arxiv.org/abs/2304.13356} {arXiv:2304.13356} \BibitemShut {NoStop}%
\bibitem [{\citenamefont {Haag}\ \emph {et~al.}(1970)\citenamefont {Haag}, \citenamefont {Kadison},\ and\ \citenamefont {Kastler}}]{HaagEtAl70}%
  \BibitemOpen
  \bibfield  {author} {\bibinfo {author} {\bibfnamefont {R.}~\bibnamefont {Haag}}, \bibinfo {author} {\bibfnamefont {R.~V.}\ \bibnamefont {Kadison}},\ and\ \bibinfo {author} {\bibfnamefont {D.}~\bibnamefont {Kastler}},\ }\bibfield  {title} {\bibinfo {title} {{Nets of C$^*$-algebras and classification of states}},\ }\href {https://doi.org/10.1007/BF01646615} {\bibfield  {journal} {\bibinfo  {journal} {Commun. Math. Phys.}\ }\textbf {\bibinfo {volume} {16}},\ \bibinfo {pages} {81} (\bibinfo {year} {1970})}\BibitemShut {NoStop}%
\bibitem [{\citenamefont {Penrose}(1972)}]{Penrose72}%
  \BibitemOpen
  \bibfield  {author} {\bibinfo {author} {\bibfnamefont {R.}~\bibnamefont {Penrose}},\ }\href@noop {} {\emph {\bibinfo {title} {{Techniques of differential topology in relativity}}}},\ \bibinfo {series} {CBMS-NSF Regional Conference Series in Applied Mathematics}, Vol.~\bibinfo {volume} {7}\ (\bibinfo  {publisher} {Society for Industrial \& Applied Mathematics},\ \bibinfo {address} {Philadelphia},\ \bibinfo {year} {1972})\BibitemShut {NoStop}%
\bibitem [{\citenamefont {Fewster}\ and\ \citenamefont {Rejzner}(2019)}]{FewsterR19}%
  \BibitemOpen
  \bibfield  {author} {\bibinfo {author} {\bibfnamefont {C.~J.}\ \bibnamefont {Fewster}}\ and\ \bibinfo {author} {\bibfnamefont {K.}~\bibnamefont {Rejzner}},\ }\href@noop {} {\bibinfo {title} {{Algebraic Quantum Field Theory - an introduction}}} (\bibinfo {year} {2019}),\ \Eprint {https://arxiv.org/abs/1904.04051} {arXiv:1904.04051} \BibitemShut {NoStop}%
\bibitem [{\citenamefont {Streater}\ and\ \citenamefont {Wightman}(2000)}]{StreaterW00}%
  \BibitemOpen
  \bibfield  {author} {\bibinfo {author} {\bibfnamefont {R.~F.}\ \bibnamefont {Streater}}\ and\ \bibinfo {author} {\bibfnamefont {A.~S.}\ \bibnamefont {Wightman}},\ }\href@noop {} {\emph {\bibinfo {title} {{PCT, Spin and Statistics, and All That}}}},\ Princeton Landmarks in Physics\ (\bibinfo  {publisher} {Princeton University Press},\ \bibinfo {year} {2000})\BibitemShut {NoStop}%
\bibitem [{\citenamefont {Baez}(1987)}]{Baez87}%
  \BibitemOpen
  \bibfield  {author} {\bibinfo {author} {\bibfnamefont {J.}~\bibnamefont {Baez}},\ }\bibfield  {title} {\bibinfo {title} {{Bell's inequality for C*-algebras}},\ }\href {https://doi.org/10.1007/BF00955201} {\bibfield  {journal} {\bibinfo  {journal} {Lett. Math. Phys.}\ }\textbf {\bibinfo {volume} {13}},\ \bibinfo {pages} {135} (\bibinfo {year} {1987})}\BibitemShut {NoStop}%
\bibitem [{\citenamefont {Summers}\ and\ \citenamefont {Werner}(1987{\natexlab{a}})}]{SummersW87a}%
  \BibitemOpen
  \bibfield  {author} {\bibinfo {author} {\bibfnamefont {S.~J.}\ \bibnamefont {Summers}}\ and\ \bibinfo {author} {\bibfnamefont {R.}~\bibnamefont {Werner}},\ }\bibfield  {title} {\bibinfo {title} {{Bell’s inequalities and quantum field theory. I. General setting}},\ }\href {https://doi.org/10.1063/1.527733} {\bibfield  {journal} {\bibinfo  {journal} {J. Math. Phys.}\ }\textbf {\bibinfo {volume} {28}},\ \bibinfo {pages} {2440} (\bibinfo {year} {1987}{\natexlab{a}})}\BibitemShut {NoStop}%
\bibitem [{\citenamefont {Clifton}\ and\ \citenamefont {Halvorson}(2001)}]{CliftonH01}%
  \BibitemOpen
  \bibfield  {author} {\bibinfo {author} {\bibfnamefont {R.}~\bibnamefont {Clifton}}\ and\ \bibinfo {author} {\bibfnamefont {H.}~\bibnamefont {Halvorson}},\ }\bibfield  {title} {\bibinfo {title} {Entanglement and open systems in algebraic quantum field theory},\ }\href {https://doi.org/https://doi.org/10.1016/S1355-2198(00)00033-2} {\bibfield  {journal} {\bibinfo  {journal} {Stud. Hist. Philos. Sci. B}\ }\textbf {\bibinfo {volume} {32}},\ \bibinfo {pages} {1} (\bibinfo {year} {2001})}\BibitemShut {NoStop}%
\bibitem [{\citenamefont {Fewster}\ and\ \citenamefont {Verch}(2013)}]{FewsterV13}%
  \BibitemOpen
  \bibfield  {author} {\bibinfo {author} {\bibfnamefont {C.~J.}\ \bibnamefont {Fewster}}\ and\ \bibinfo {author} {\bibfnamefont {R.}~\bibnamefont {Verch}},\ }\bibfield  {title} {\bibinfo {title} {{The necessity of the Hadamard condition}},\ }\href {https://doi.org/10.1088/0264-9381/30/23/235027} {\bibfield  {journal} {\bibinfo  {journal} {Class. Quantum Grav.}\ }\textbf {\bibinfo {volume} {30}},\ \bibinfo {pages} {235027} (\bibinfo {year} {2013})}\BibitemShut {NoStop}%
\bibitem [{\citenamefont {Summers}\ and\ \citenamefont {Werner}(1987{\natexlab{b}})}]{SummersW87b}%
  \BibitemOpen
  \bibfield  {author} {\bibinfo {author} {\bibfnamefont {S.~J.}\ \bibnamefont {Summers}}\ and\ \bibinfo {author} {\bibfnamefont {R.}~\bibnamefont {Werner}},\ }\bibfield  {title} {\bibinfo {title} {{Bell’s inequalities and quantum field theory. II. Bell’s inequalities are maximally violated in the vacuum}},\ }\href {https://doi.org/10.1063/1.527734} {\bibfield  {journal} {\bibinfo  {journal} {J. Math. Phys.}\ }\textbf {\bibinfo {volume} {28}},\ \bibinfo {pages} {2448} (\bibinfo {year} {1987}{\natexlab{b}})}\BibitemShut {NoStop}%
\bibitem [{\citenamefont {Buchholz}\ \emph {et~al.}(1986)\citenamefont {Buchholz}, \citenamefont {Doplicher},\ and\ \citenamefont {Longo}}]{BuchholzEtAl86}%
  \BibitemOpen
  \bibfield  {author} {\bibinfo {author} {\bibfnamefont {D.}~\bibnamefont {Buchholz}}, \bibinfo {author} {\bibfnamefont {S.}~\bibnamefont {Doplicher}},\ and\ \bibinfo {author} {\bibfnamefont {R.}~\bibnamefont {Longo}},\ }\bibfield  {title} {\bibinfo {title} {On noether's theorem in quantum field theory},\ }\href {https://doi.org/https://doi.org/10.1016/0003-4916(86)90086-2} {\bibfield  {journal} {\bibinfo  {journal} {Ann. Phys.}\ }\textbf {\bibinfo {volume} {170}},\ \bibinfo {pages} {1} (\bibinfo {year} {1986})}\BibitemShut {NoStop}%
\bibitem [{\citenamefont {Werner}(1987)}]{Werner87}%
  \BibitemOpen
  \bibfield  {author} {\bibinfo {author} {\bibfnamefont {R.}~\bibnamefont {Werner}},\ }\bibfield  {title} {\bibinfo {title} {Local preparability of states and the split property in quantum field theory},\ }\href {https://doi.org/10.1007/BF00401161} {\bibfield  {journal} {\bibinfo  {journal} {Lett. Math. Phys.}\ }\textbf {\bibinfo {volume} {13}},\ \bibinfo {pages} {325} (\bibinfo {year} {1987})}\BibitemShut {NoStop}%
\bibitem [{\citenamefont {Kraus}\ \emph {et~al.}(1983)\citenamefont {Kraus}, \citenamefont {B{\"o}hm}, \citenamefont {Dollard},\ and\ \citenamefont {Wootters}}]{Kraus83}%
  \BibitemOpen
  \bibfield  {author} {\bibinfo {author} {\bibfnamefont {K.}~\bibnamefont {Kraus}}, \bibinfo {author} {\bibfnamefont {A.}~\bibnamefont {B{\"o}hm}}, \bibinfo {author} {\bibfnamefont {J.}~\bibnamefont {Dollard}},\ and\ \bibinfo {author} {\bibfnamefont {W.}~\bibnamefont {Wootters}},\ }\href@noop {} {\emph {\bibinfo {title} {States, Effects, and Operations: Fundamental Notions of Quantum Theory}}},\ Lecture Notes in Physics\ (\bibinfo  {publisher} {Springer Berlin Heidelberg},\ \bibinfo {year} {1983})\BibitemShut {NoStop}%
\bibitem [{\citenamefont {Ozawa}(2012)}]{Ozawa12}%
  \BibitemOpen
  \bibfield  {author} {\bibinfo {author} {\bibfnamefont {M.}~\bibnamefont {Ozawa}},\ }\href@noop {} {\bibinfo {title} {{Mathematical foundations of quantum information: Measurement and foundations}}} (\bibinfo {year} {2012}),\ \Eprint {https://arxiv.org/abs/1201.5334} {arXiv:1201.5334} \BibitemShut {NoStop}%
\bibitem [{\citenamefont {Okamura}\ and\ \citenamefont {Ozawa}(2015)}]{OkamuraO15}%
  \BibitemOpen
  \bibfield  {author} {\bibinfo {author} {\bibfnamefont {K.}~\bibnamefont {Okamura}}\ and\ \bibinfo {author} {\bibfnamefont {M.}~\bibnamefont {Ozawa}},\ }\bibfield  {title} {\bibinfo {title} {{Measurement theory in local quantum physics}},\ }\href {https://doi.org/10.1063/1.4935407} {\bibfield  {journal} {\bibinfo  {journal} {J. Math. Phys.}\ }\textbf {\bibinfo {volume} {57}},\ \bibinfo {pages} {015209} (\bibinfo {year} {2015})}\BibitemShut {NoStop}%
\bibitem [{\citenamefont {Bostelmann}\ \emph {et~al.}(2021)\citenamefont {Bostelmann}, \citenamefont {Fewster},\ and\ \citenamefont {Ruep}}]{BostelmannEtAl21}%
  \BibitemOpen
  \bibfield  {author} {\bibinfo {author} {\bibfnamefont {H.}~\bibnamefont {Bostelmann}}, \bibinfo {author} {\bibfnamefont {C.~J.}\ \bibnamefont {Fewster}},\ and\ \bibinfo {author} {\bibfnamefont {M.~H.}\ \bibnamefont {Ruep}},\ }\bibfield  {title} {\bibinfo {title} {{Impossible measurements require impossible apparatus}},\ }\href {https://doi.org/10.1103/PhysRevD.103.025017} {\bibfield  {journal} {\bibinfo  {journal} {Phys. Rev. D}\ }\textbf {\bibinfo {volume} {103}},\ \bibinfo {pages} {025017} (\bibinfo {year} {2021})}\BibitemShut {NoStop}%
\bibitem [{\citenamefont {Pranzini}\ \emph {et~al.}(2022)\citenamefont {Pranzini}, \citenamefont {García-Pérez}, \citenamefont {Keski-Vakkuri},\ and\ \citenamefont {Maniscalco}}]{PranziniEtAl23}%
  \BibitemOpen
  \bibfield  {author} {\bibinfo {author} {\bibfnamefont {N.}~\bibnamefont {Pranzini}}, \bibinfo {author} {\bibfnamefont {G.}~\bibnamefont {García-Pérez}}, \bibinfo {author} {\bibfnamefont {E.}~\bibnamefont {Keski-Vakkuri}},\ and\ \bibinfo {author} {\bibfnamefont {S.}~\bibnamefont {Maniscalco}},\ }\href@noop {} {\bibinfo {title} {{Born rule extension for non-replicable systems and its consequences for Unruh-DeWitt detectors}}} (\bibinfo {year} {2022}),\ \Eprint {https://arxiv.org/abs/2210.13347} {arXiv:2210.13347} \BibitemShut {NoStop}%
\bibitem [{\citenamefont {de~Ram\'on}\ \emph {et~al.}(2018)\citenamefont {de~Ram\'on}, \citenamefont {Garay},\ and\ \citenamefont {Mart\'{\i}n-Mart\'{\i}nez}}]{deRamonEtAl18}%
  \BibitemOpen
  \bibfield  {author} {\bibinfo {author} {\bibfnamefont {J.}~\bibnamefont {de~Ram\'on}}, \bibinfo {author} {\bibfnamefont {L.~J.}\ \bibnamefont {Garay}},\ and\ \bibinfo {author} {\bibfnamefont {E.}~\bibnamefont {Mart\'{\i}n-Mart\'{\i}nez}},\ }\bibfield  {title} {\bibinfo {title} {Direct measurement of the two-point function in quantum fields},\ }\href {https://doi.org/10.1103/PhysRevD.98.105011} {\bibfield  {journal} {\bibinfo  {journal} {Phys. Rev. D}\ }\textbf {\bibinfo {volume} {98}},\ \bibinfo {pages} {105011} (\bibinfo {year} {2018})}\BibitemShut {NoStop}%
\bibitem [{\citenamefont {Derezi{\'{n}}ski}\ and\ \citenamefont {Meissner}(2006)}]{DerezinskiM06}%
  \BibitemOpen
  \bibfield  {author} {\bibinfo {author} {\bibfnamefont {J.}~\bibnamefont {Derezi{\'{n}}ski}}\ and\ \bibinfo {author} {\bibfnamefont {K.~A.}\ \bibnamefont {Meissner}},\ }\bibinfo {title} {Quantum massless field in 1+1 dimensions},\ in\ \href {https://doi.org/10.1007/3-540-34273-7_11} {\emph {\bibinfo {booktitle} {Mathematical Physics of Quantum Mechanics: Selected and Refereed Lectures from QMath9}}},\ \bibinfo {editor} {edited by\ \bibinfo {editor} {\bibfnamefont {J.}~\bibnamefont {Asch}}\ and\ \bibinfo {editor} {\bibfnamefont {A.}~\bibnamefont {Joye}}}\ (\bibinfo  {publisher} {Springer Berlin Heidelberg},\ \bibinfo {address} {Berlin, Heidelberg},\ \bibinfo {year} {2006})\ pp.\ \bibinfo {pages} {107--127}\BibitemShut {NoStop}%
\bibitem [{\citenamefont {Streater}\ and\ \citenamefont {Wilde}(1970)}]{StreaterW70}%
  \BibitemOpen
  \bibfield  {author} {\bibinfo {author} {\bibfnamefont {R.~F.}\ \bibnamefont {Streater}}\ and\ \bibinfo {author} {\bibfnamefont {I.~F.}\ \bibnamefont {Wilde}},\ }\bibfield  {title} {\bibinfo {title} {{Fermion states of a boson field}},\ }\href {https://doi.org/10.1016/0550-3213(70)90445-1} {\bibfield  {journal} {\bibinfo  {journal} {Nucl. Phys. B}\ }\textbf {\bibinfo {volume} {24}},\ \bibinfo {pages} {561} (\bibinfo {year} {1970})}\BibitemShut {NoStop}%
\end{thebibliography}%

\newpage
\onecolumngrid
\appendix
\section{\texorpdfstring{Proof of Eq.~\eqref{e.FactoringMeansIdentity}}{Proof of Eq. (39)}}
\label{a.ProofFactoringMeansIdentity}
\noindent In order to prove Eq.~\eqref{e.FactoringMeansIdentity}, let us assign the detector a Hilbert space $\mathcal{H}_D$ of dimension $d$ (in the main text, this was chosen to be $2$). Hence, the state after the (entangling) interaction can be expanded as
\begin{equation}
    \ket{\psi}=\sum_{i=1}^d\gamma_i\ket{i}\ket{\phi_i}
\end{equation}
for some set of QFT states $\{\ket{\phi_i}\}$ such that $\braket{\phi_i|\phi_j}\neq 1$, where at least two of the $\gamma_i$ coefficient are non-zero, and they satisfy
\begin{equation}
    \sum_{i=1}^d|\gamma_i|^2=1~.
\end{equation}
Then, we can expand the LHS of Eq.~\eqref{e.FactoringMeansIdentity} as
\begin{equation}
    \Tr[\rho' \hat{L}\hat{\Gamma}]=\bra{\psi}\hat{L}\hat{\Gamma}\ket{\psi}=\sum_{ij}\gamma_i^*\gamma_j\bra{\phi_i}\hat{L}\ket{\phi_j}\bra{i}\hat{\Gamma}\ket{j}~,
    \label{e.app1.1}
\end{equation}
and, in the same way, obtain the RHS as
\begin{equation}
    \Tr[\rho' \hat{L}]\Tr[\rho' \hat{\Gamma}]=\sum_i|\gamma_i|^2\bra{\phi_i}\hat{L}\ket{\phi_i}\times\sum_{ij}\gamma_i^*\gamma_j\braket{\phi_i|\phi_j}\bra{i}\hat{\Gamma}\ket{j}~.
    \label{e.app1.2}
    \end{equation}
Since Eq.s~\eqref{e.app1.1} and \eqref{e.app1.2} must be equal for all $\hat{L}$, we get that
\begin{equation}
    \begin{dcases}
    \bra{i}\hat{\Gamma}\ket{i}=\sum_{ij}\gamma_i^*\gamma_j\braket{\phi_i|\phi_j}\bra{i}\hat{\Gamma}\ket{j}\\
    \gamma^*_i\gamma_j\bra{i}\hat{\Gamma}\ket{j}=0
    \end{dcases}
\end{equation}
By the last $d(d-1)/2$ equations, we get that $\hat{\Gamma}$ must be diagonal in the $\{\ket{i}\}$ basis; by using this information in the first $d$ equations we get
\begin{equation}
    \bra{i}\hat{\Gamma}\ket{i} =\sum_j |\gamma_j|^2\bra{j}\hat{\Gamma}\ket{j}~;
\end{equation}
for $\hat{\Gamma}$ to be independent of the state components, we must have
\begin{equation}
    \bra{i}\hat{\Gamma}\ket{i}=\bra{j}\hat{\Gamma}\ket{j}~~\forall i,j~,
\end{equation}
meaning that
\begin{equation}
    \hat{\Gamma}\propto \mathbb{I}~.
\end{equation}

\section{\texorpdfstring{Proof of Eq.~\eqref{e.9isUnique}}{Proof of Eq. (42)}}
\label{a.proof5689}
\noindent The proof of Eq.~\eqref{e.9isUnique} is split into two parts. First, we prove that
\begin{equation}
    \langle\hat{L}\rangle_{\bf 9}\neq \langle\hat{L}\rangle_{\bf 1,\dots,4}~;
    \label{b.first part}
\end{equation}
second, we prove that 
\begin{equation}
    \langle\hat{L}\rangle_{\bf 9}\neq \langle\hat{L}\rangle_{\bf 5,\dots,8}~.
    \label{b.second part}
\end{equation}
\subsection{\texorpdfstring{Proof of Eq.~\eqref{b.first part}}{Proof of Eq. (B1)}}
\noindent In order to prove Eq.~\eqref{b.first part}, we consider the expectation values of any local $\hat{L}$ having support on $\mathcal{P}^+_A\cap\mathcal{P}^+_B$ over $\rho_{\bf 4}$, i.e.
\begin{equation}
    \langle\hat{L}\rangle_{\bf 4}=\Tr[\hat{U}_A\hat{U}_B(\rho_{\phi_0}\otimes\rho_{\alpha}\otimes\sigma_\beta)\hat{U}^\dagger_B\hat{U}_A^\dagger\hat{L}]~.
\end{equation}
By defining $\rho'=\hat{U}_A\hat{U}_B(\rho_{\phi_0}\otimes\rho_{\alpha}\otimes\sigma_\beta)\hat{U}^\dagger_B\hat{U}_A^\dagger$, and $\hat{\Gamma}_a=\hat{M}^\dagger_a\hat{M}_a$ and $\hat{\Gamma}_b=\hat{M}^\dagger_b\hat{M}_b$, we have that
\begin{equation}
    \langle\hat{L}\rangle_{\bf 4}=\Tr[\rho'\hat{L}]
\end{equation}
and
\begin{equation}
    \langle\hat{L}\rangle_{\bf 9}=\frac{\Tr[\rho'\hat{\Gamma}_a\hat{\Gamma}_b\hat{L}]}{\Tr[\rho'\hat{\Gamma}_a\hat{\Gamma}_b]}~.
\end{equation}
Therefore, proceeding \textit{ad absurdum} and negating Eq.~\eqref{b.first part} gives
\begin{equation}
    \Tr[\rho'\hat{\Gamma}_a\hat{\Gamma}_b\hat{L}]=\Tr[\rho'\hat{\Gamma}_a\hat{\Gamma}_b]\Tr[\rho'\hat{L}]
    \label{b.absurdum1}
\end{equation}
which, by applying the same procedure as in App.~\ref{a.ProofFactoringMeansIdentity}, gives
\begin{equation}
    \hat{\Gamma}_a\hat{\Gamma}_b\propto\mathbb{I}~\Rightarrow~\hat{\Gamma}_{a}\propto\mathbb{I}_{A}~{\rm and}~\hat{\Gamma}_{b}\propto\mathbb{I}_{B}~.
\end{equation}
However, when building the measurements we selected the measurement operator to satisfy $\hat{M}^\dagger_a\hat{M}_a\neq k\mathbb{I}_A$ (and similarly for $\hat{M}_b$), where $k$ is some proportionality constant so that they described a non-trivial measurement. The contradiction means that Eq.~\eqref{b.absurdum1} is false for non-trivial measurements.

\subsection{\texorpdfstring{Proof of Eq.~\eqref{b.second part}}{Proof of Eq. (B2)}}
\noindent To prove Eq.~\eqref{b.second part}, we consider the expectation values of any local $\hat{L}$ having support on $\mathcal{P}^+_A\cap\mathcal{P}^+_B$ over $\rho_{\bf 6}$ (similarly, one can perform the proof for $\rho_{\bf 8}$). By the same notation as above, we have that
\begin{equation}
    \langle\hat{L}\rangle_{\bf 6}=\frac{\Tr[\rho'\hat{\Gamma}_a\hat{L}]}{\Tr[\rho'\hat{\Gamma}_a]}~.
\end{equation}
Therefore, by proceeding again \textit{ad absurdum} and negating
\begin{equation}
    \langle\hat{L}\rangle_{\bf 9}\neq \langle\hat{L}\rangle_{\bf 6}
\end{equation}
means that
\begin{equation}
    \Tr[\rho'\hat{\Gamma}_a\hat{\Gamma}_b\hat{L}]\Tr[\rho'\hat{\Gamma}_a]=\Tr[\rho'\hat{\Gamma}_a\hat{\Gamma}_b]\Tr[\rho'\hat{\Gamma}_a\hat{L}]~.
    \label{b.absurdum2}
\end{equation}
Showing that this implies that the measurement operators are inconsistent with a good measurement choice is harder. To do so, let us assume without loss of generality that $dim(\mathcal{H}_A)\leq dim(\mathcal{H}_B)$, and write the states appearing in $\rho'=\ket{\psi}\bra{\psi}$ as
\begin{equation}
    \ket{\psi}=\sum_i\gamma_i\ket{i}\ket{\beta_i}\ket{\phi_i}~.
\end{equation}
Using this expression, it is possible to rewrite the LHS of Eq.~\eqref{b.absurdum2} as
\begin{equation}
    \sum_{ij}\gamma^*_i\gamma_j\bra{i}\hat{\Gamma}_a\ket{j}\braket{\beta_i|\beta_j}\braket{\phi_i|\phi_j} \times \sum_{ij}\gamma^*_i\gamma_j\bra{i}\hat{\Gamma}_a\ket{j}\bra{\beta_i}\hat{\Gamma}_b\ket{\beta_j}\bra{\phi_i}\hat{L}\ket{\phi_j}
\end{equation}
and its RHS as
\begin{equation}
    \sum_{ij}\gamma^*_i\gamma_j\bra{i}\hat{\Gamma}_a\ket{j}\bra{\beta_i}\hat{\Gamma}_b\ket{\beta_j}\braket{\phi_i|\phi_j} \times \sum_{ij}\gamma^*_i\gamma_j\bra{i}\hat{\Gamma}_a\ket{j}\braket{\beta_i|\beta_j}\bra{\phi_i}\hat{L}\ket{\phi_j}~;
\end{equation}
matching these for all $\hat{L}$ means that
\begin{equation}
    \bra{\beta_m}\hat{\Gamma}_b\ket{\beta_n}\sum_{ij}\gamma^*_i\gamma_j\bra{i}\hat{\Gamma}_a\ket{j}\braket{\beta_i|\beta_j}\braket{\phi_i|\phi_j} =\braket{\beta_m|\beta_n}\sum_{ij}\gamma^*_i\gamma_j\bra{i}\hat{\Gamma}_a\ket{j}\bra{\beta_i}\hat{\Gamma}_b\ket{\beta_j}\braket{\phi_i|\phi_j}
\end{equation}
for all $m,n$. Moreover, since the properties of $\hat{\Gamma}_b$ cannot depend on those of $\hat{\Gamma}_a$, we can select the latter to be any effect of our choice. In particular, by choosing it to be a projector over $\ket{k}$, we get
\begin{equation}
    \bra{\beta_m}\hat{\Gamma}_b\ket{\beta_n}=\braket{\beta_m|\beta_n}\bra{\beta_k}\hat{\Gamma}_b\ket{\beta_k}~.
\end{equation}
By choosing $m=n$ we obtain that
\begin{equation}
    \bra{\beta_m}\hat{\Gamma}_b\ket{\beta_m}=\bra{\beta_k}\hat{\Gamma}_b\ket{\beta_k}
\end{equation}
for all $k$ and $m$. By calling $\Gamma$ these elements we get
\begin{equation}
    \bra{\beta_m}\hat{\Gamma}_b\ket{\beta_n}=\braket{\beta_m|\beta_n}\Gamma~.
\end{equation}
which, since it has to be true for all $m$ and $n$ and for all possible entangling interactions, means that
\begin{equation}
    \hat{\Gamma}_b\propto\mathbb{I}~.
\end{equation}
As mentioned above, we reach a contradiction and hence our proof is complete.

\section{The causal structure around measurement regions}

\label{a.open_causally_convex_set}
In this section, we prove that a causally convex open region enclosing two points that lie out a given measurement region and its causal past/future always exists in $d>2$. To obtain this result, we first need to prove some lemmata.
\begin{lemma}
    Consider two points $p,q\in\mathcal{M}$ such that $p\notin \mathcal{J}(q)$. Then there exists $\mathcal{O}_p$ open neighbourood of $p$ such that  $\mathcal{D}(\mathcal{O}_p)\cap\mathcal{J}(p)=\emptyset$.
\end{lemma}
\begin{proof}
    First we notice that there always exists $\mathcal{O}_p$ open neighbourood of $p$ such that $\mathcal{O}_p\cap \mathcal{J}(q)=\emptyset$. Then, we proceed ad absurdum and suppose $\mathcal{D}(\mathcal{O}_p)\cap\mathcal{J}(q)\neq\emptyset$. Then $\exists x\in \mathcal{J}(q)$ such that $x\in\mathcal{D}(\mathcal{O}_p)$, which means that either $x\in\mathcal{O}_p$ or $x\in\mathcal{D}(\mathcal{O}_p)\setminus\mathcal{O}_p$. In the first case, $x\in\mathcal{O}_p\cap \mathcal{J}(q)\neq\emptyset$; in the second case, $x\in\mathcal{D}(\mathcal{O}_p)\setminus\mathcal{O}_p$ means $\exists \gamma:\mathbb{R}\rightarrow\mathcal{M}$ timelike curve passing through $x$ such that $\gamma(0)\in\mathcal{O}_p$ and $\gamma(1)=q$, which implies $\gamma(0)\in\mathcal{J}(q)$ and hence $\gamma(0)\in\mathcal{O}_p\cap \mathcal{J}(q)\neq\emptyset$. Both conclusions are contrary to the hypothesis.
\end{proof}
\begin{lemma}
    For any $M\subset\mathcal{M}$ and curve $\gamma:[0,1]\rightarrow\mathcal{M}$ such that $\gamma(\sigma)\notin \mathcal{J}(M),$ $\forall\sigma\in[0,1]$, then 
    \begin{equation}
        \left(\cup_{\sigma\in[0,1]}\mathcal{D}(\mathcal{O}_{\gamma(\sigma)})\right)\cap\mathcal{J}(M)=\emptyset~,
    \end{equation}
    where $\mathcal{O}_{\gamma(\sigma)}$ is some open neighbourhood of $\gamma(\sigma)$. Moreover, $\mathcal{D}_\gamma\equiv\cup_{\sigma\in[0,1]}\mathcal{D}(\mathcal{O}_{\gamma(\sigma)})$ is causally convex.
\end{lemma}
\begin{proof}
    We just use lemma 1 and the topological definition of open set. Then, $\mathcal{D}_\gamma$ is causally convex because it is a union of causally convex sets.
\end{proof}

We are now ready to prove our theorem:
\begin{theorem}
    For any measurement region $M$ defined on a Minkowski spacetime of dimension $d>2$ and two points $x,y\notin J(M)$, there exists a bounded open causally convex region $R$  such that $x,y\in R$ and $R\cap \mathcal{J}(M)=\emptyset$.
\end{theorem}
\begin{proof}
    First, we notice that there always exist open neighbourhoods $\mathcal{O}_x$ and $\mathcal{O}_y$, respectively of $x$ and $y$, that are disjoint from $\mathcal{J}(M)$, because $x,y\notin \mathcal{J}(M)$ and $M$ is closed. Next, we distinguish the two cases for which the points are either 1) timelike or lightlike, or 2) spacelike separated. 

In the first case, we pick as $R$ the one of the two regions $R^{\pm}=\mathcal{J}^\mp(\mathcal{O}_x)\cap\mathcal{J}^\pm(\mathcal{O}_y)$ that is not empty; without loss of generality, we assume $R=R^+$. Then, $\mathcal{O}_x\cap\mathcal{J}^+(M)=\emptyset$ implies $\mathcal{J}^-(\mathcal{O}_x)\cap \mathcal{J}^+(M)=\emptyset$, and $\mathcal{O}_y\cap\mathcal{J}^-(M)=\emptyset$ implies $\mathcal{J}^+(\mathcal{O}_y)\cap \mathcal{J}^-(M)=\emptyset$. Hence $R\cap\mathcal{J}(M)=\emptyset$.

Next, we consider $x$ and $y$ spacelike. Since they are also spacelike to any $p\in M$, we can pick a spacelike slice $\Sigma$ that includes $x, y$ and some $p$. Next, we define a curve $\gamma:[0,1]\rightarrow \Sigma$ with $\gamma(0)=x$, $\gamma(1)=y$, and such that $\gamma\subset \Sigma'$, where $\Sigma'= \Sigma \setminus \mathcal{J}(M)$. Defining such a curve is possible when $\Sigma'$ is path-connected, which is the case because $\Sigma'$ is homeomorphic to $\mathbb{R}^{d-1}\setminus \{0\}$ which, since $\mathcal{M}$ has dimension $d>2$, is always path-connected. Once the curve has been defined, we can use Lemma 2 to find a causally convex region enclosing it.
\end{proof}

Notice that the requirement that the spacetime be Minkowski and have dimension larger than two can be relaxed by choosing non-trivial spacetime topologies. For example, punctured spacelike slices of a $1+1$ spacetime with periodic spacelike dimension (such as the Lorentzian cylinder $\mathbb{R}\times S^1$) are also path-connected. In these non-trivial topologies, one might be able to recover the vacuum-like description of the post-measurement state outside $\mathcal{J}^+(M)$; in doing so, one should keep in mind the adage saying that ``\textit{quantum massless bosonic fields in 1+1 dimensions do not exist}''\cite{DerezinskiM06} and hence be careful in constructing the states of the theory, keeping in mind that even constructing a good notion of the ground state in two-dimensional free field theories with noncompact Cauchy surfaces is problematic due to IR divergences~\cite{StreaterW70}.

\end{document}